\documentclass[lettersize,journal]{IEEEtran}
\usepackage{amsmath,amsfonts}
\usepackage{algorithmic}
\usepackage{array}
\usepackage{textcomp}
\usepackage{stfloats}
\usepackage{url}
\usepackage{hyperref}
\usepackage{verbatim}
\usepackage{graphicx}
\usepackage{algorithm}
\usepackage{tikz}
\usetikzlibrary{arrows.meta,quotes,calc,angles,patterns,decorations.pathmorphing,decorations.markings}
\newtheorem{theorem}{Theorem}[section]
\newtheorem{proof}{Proof}[section]
\newtheorem{assumption}{Assumption}
\def\BibTeX{{\rm B\kern-.05em{\sc i\kern-.025em b}\kern-.08em
    T\kern-.1667em\lower.7ex\hbox{E}\kern-.125emX}}
\usepackage{balance}
\begin{document}
\title{A Class of Hierarchical Sliding Mode Control based on Extended Kalman filter for Quadrotor UAVs}
\author{Van Chung Nguyen and  Hung Manh La % <-this % stops a space
\thanks{This work was funded by the U.S. National Science Foundation (NSF) under grants NSF-CAREER: 1846513 and NSF-PFI-TT: 1919127, and the U.S. Department of Transportation, Office of the Assistant Secretary for Research and Technology (USDOT/OST-R) under Grant No. 69A3551747126 through INSPIRE University Transportation Center. The views, opinions, findings, and conclusions reflected in this publication are solely those of the authors and do not represent the official policy or position of 
 the NSF, the USDOT/OST-R, or any other entities.}
\thanks{The authors are with the Advanced Robotics and Automation (ARA) Lab, Department of Computer Science and Engineering, University of Nevada, Reno, NV 89557, USA.}
\thanks{$^*$ \url{https://github.com/aralab-unr/HSMC-EKF-for-Quadrotor-UAVs}}
\thanks{Corresponding author: Hung La, email: {\tt\small hla@unr.edu}}%
}
\markboth{}
{}
\maketitle
\begin{abstract}
This study introduces a novel methodology for controlling Quadrotor Unmanned Aerial Vehicles (UAVs), focusing on Hierarchical Sliding Mode Control (HSMC) strategies and an Extended Kalman Filter (EKF). Initially, an EKF is proposed to enhance robustness in estimating UAV states, thereby reducing the impact of measured noises and external disturbances. By locally linearizing UAV systems, the EKF can mitigate the disadvantages of the Kalman filter and reduce the computational cost of other nonlinear observers. Subsequently, in comparison to other related work in terms of stability and computational cost, the HSMC framework shows its outperformance in allowing the quadrotor UAVs to track the references. Three types of HSMC—Aggregated HSMC (AHSMC), Incremental HSMC (IHSMC), and Combining HSMC (CHSMC)—are investigated for their effectiveness in tracking reference trajectories. Moreover, the stability of the quadrotor UAVs is rigorously analyzed using the Lyapunov stability principle. Finally, experimental results and comparative analyses demonstrate the efficacy and feasibility of the proposed methodologies.
\end{abstract}
\def\abstractname{Note to Practitioners}
\begin{abstract}
This paper was motivated by the problem of stability in controlling Quadrotor UAVs under the influence of measurement noise and external disturbances. Existing approaches to control UAVs typically consider only the position loop as a single z-axis position. By proposing a novel structure, this paper focuses on presenting a robust approach to control Quadrotor UAVs concerning Lyapunov stability. Firstly, the EKF is conducted by linearizing the system around the current working status and estimating the system's states using current feedback data. Secondly, the attitude loop and position loop are handled by a group of PD-SMC controllers and HSMC, respectively. The reference roll and pitch angles are calculated by the PD, while the SMC ensures the tracking of the Quadrotor UAVs. The position controller is managed by three types of HSMC, each considering a differential structure with specific control laws to ensure the tracking stability of the UAVs. The proposed strategy, with its specific control laws, reduces the computational burden compared to other intelligent methods and observers, making it suitable for real-time applications in resource-constrained environments. The rigorous analysis based on the Lyapunov stability principle ensures that the proposed control strategies maintain the stability of Quadrotor UAVs under various scenarios. In future research, we plan to utilize the proposed Lyapunov function to develop adaptive controllers that address system uncertainty and fault tolerance.
\end{abstract}
\section{INTRODUCTION}
The pervasive integration of Unmanned Aerial Vehicles (UAVs) across diverse sectors, including aerial photography, agriculture, disaster response, and surveillance, has created many new opportunities in modern industry. This technological advancement has revolutionized operational methodologies and attracted immense interest and engagement within both industrial and academic spheres. In addition, the utilization of UAVs in real life creates many advantages for humans ranging from the possibility of removing human pilots from danger and reducing the size and the cost of the UAVs \cite{besnard2012quadrotor}. Therefore, the design of controllers that allow the UAVs to perform accurately and effectively is necessary. Recently, numerous review papers \cite{han2018review,amin2016review} have highlighted the extensive research conducted in this domain, underscoring the diverse approaches explored. Specifically, the controller for quadrotor UAVs encompasses three primary categories: Linear controllers, Nonlinear controllers, and Intelligent controllers.
While linear controllers, such as the PID controller referenced in \cite{pham2018autonomous, bouaiss2020modeling}, or the LQR controller mentioned in \cite{elkhatem2022robust, rinaldi2013linear}, can assist UAVs in tracking the desired trajectory, they often suffer from low performance. Additionally, linear controllers's robustness is contingent on various factors when they are highly sensitive to external disturbances and model uncertainties. Conversely, intelligent controllers, including artificial neural networks as discussed in \cite{hay2023noise, sadhu2023board}, or reinforcement learning approaches as seen in \cite{koch2019reinforcement, wang2019autonomous}, provide better solutions for enhanced trajectory tracking and improved robustness for UAVs. In addition, some intelligent controllers do not require the model of the UAVs such as in \cite{muller2023robust}. Nevertheless, the application of intelligent controllers demands substantial data for learning and incurs high costs when applied to UAV systems. The comparison of the proposed method with the other related works is provided in Table~\ref{comparision method}.
\begin{table*}
\centering
\caption{QUADROTOR UAVS' CONTROLLERS COMPARISON. SYMBOL: ADDRESSED $\surd$ , NOT ADDRESSED $\square$}
{\begin{tabular}{c c c c c c c c} \hline
\bf{Method}&\bf{Tracking ability}&\bf{Lyapunov Stability} & \bf{Noises \& Disturbances}  &\bf{Computational burden} & \bf{Available sources}  \\ \hline
PID  \cite{pham2018autonomous, bouaiss2020modeling} & $\surd$ & $\square$ & $\square$ & Low & $\square$ \\ 
LQR \cite{elkhatem2022robust, rinaldi2013linear}  & $\surd$ & $\square$ & $\square$ & Low & $\square$ \\ 
Backstepping \cite{zhang2019robust,koksal2020backstepping}   & $\surd$ & $\surd$ & $\square$ & Low & $\square$  \\ 
SMC \cite{zheng2014second,shao2021adaptive}   & $\surd$ & $\surd$ & $\square$ & Low & $\square$  \\ 
PF \cite{ma2021artificial,woods2017novel}   & $\surd$ & $\square$ & $\square$ & Low & $\square$  \\ 
MPC \cite{zhang2021robust,cavanini2021model}  & $\surd$ & $\square$ & $\square$ & High & $\square$ \\ 
Intelligent control \cite{hay2023noise, sadhu2023board,muller2023robust} & $\surd$ & $\square$ & $\surd$ & High & $\square$  \\ 
Proposed method & $\surd$ & $\surd$ & $\surd$ & Low& $\surd$\\ 
\hline
\end{tabular}}
\label{comparision method}
\end{table*}	

To overcome the challenges of tracking performance, system stability, and cost reduction, nonlinear controllers have gained widespread application in quadrotor UAVs. Some nonlinear controllers such as the backstepping controller in \cite{zhang2019robust,koksal2020backstepping}, the model predictive control (MPC) in \cite{zhang2021robust,cavanini2021model}, the sliding mode control (SMC) in \cite{zheng2014second,shao2021adaptive}, or the potential field (PF) controller in \cite{ma2021artificial,woods2017novel}  have been implemented for the UAVs with demonstrated effectiveness. However, the efficacy of nonlinear control laws depends heavily on the accuracy of the UAV model. Since UAVs are underactuated systems \cite{zheng2014second}, ensuring stability while controlling them remains a challenge for developers. In \cite{cavanini2021model,didier2021robust}, the linear parameter-varying and the tube model predictive control are proposed to control the UAVs. These strategies exhibit enhanced tracking performance when compared to conventional approaches. Nonetheless, the integration of MPC into real-mode UAV operations requires substantial computational effort. The higher the computational complexity, the greater the implementation cost. Furthermore, it is noteworthy that these control methodologies have not addressed aspects concerning UAV stability and the recursive feasibility of MPC which need further discussion. In \cite{zheng2014second}, the second-order sliding mode control is proposed for the quadrotor UAV, this controller divides the UAV system into two parts: the actuated system ($z,\psi$) and the un-actuated system ($x,y,\phi,\theta$) and then proposed the controller for each system and the stability is carried out based on the Lyapunov function. Using the Lyapunov stability also provides a promising solution for model uncertainties such as the adaptive solutions in \cite{mofid2018adaptive,labbadi2020robust,lian2021adaptive}. In \cite{mofid2018adaptive}, the internal moments and UAV's mass are estimated using the finite-time adaptive sliding mode control. The fast nonsingular terminal sliding-mode control is used to estimate the internal moments and lumped uncertainty in \cite{labbadi2020robust,lian2021adaptive} respectively. Inheriting the advantages of the Lyapunov stability method, this paper introduces a novel approach for controlling UAVs: a hierarchical sliding mode control (HSMC) framework, along with an analysis of its stability. 

Another aspect of UAVs is observers,  the UAVs need to know precisely its position, velocity, altitude, etc. In terms of the UAVs, if these states are obtained accurately then the control of the aircraft can be achieved successfully \cite{hajiyev2013robust}. Since the UAVs work in an outdoor environment, various external factors such as wing forces or measure noises influence the UAVs' performance. Therefore, proposing an adequate solution for this problem is essential in controlling UAVs.
Several observers have been developed to estimate system states, such as high-gain observers in \cite{hussain2020extended,boss2020robust}, or extended state observers, as discussed in \cite{qi2021problems,mokhtari2016extended}, which estimate both system states and disturbances of UAVs to enhance performance. In this research, we focus on the challenge of estimating UAV's system states under the influence of measured noises and external disturbances. An effective approach to address this challenge is utilizing the Kalman Filter (KF). By modeling noises as Gaussian functions, KF iteratively updates observer gain to estimate system states. However, KF is designed for linear systems, posing limitations for highly nonlinear UAV dynamics. To overcome this constraint, alternatives such as the Unscented Kalman Filter (UKF) \cite{khamseh2019unscented, de2011uav} and the Particle Filter (PF) \cite{rigatos2012nonlinear} have been proposed to track references in UAVs. Although these filters perform a high quality for the UAVs in terms of tracking performance, the computational cost of these nonlinear methods is high, leading to the rise of the cost of UAV development and operation. Inspired by the method in \cite{kallapur2007application}, by locally linearizing UAV dynamics, we propose an Extended Kalman filter (EKF) for the UAVs. This filter can ensure robust state estimation under the influence of measured noises and external disturbances. Moreover, the computational cost of EKF can be reduced in comparison to other nonlinear methods that provide potential applications to real UAVs.

In this work, we propose an Extended Kalman Filter (EKF) and two loop controllers for quadrotor UAVs. The two loop controllers include position and attitude controllers. The position controller combines a PD controller and a single sliding mode controller (PD-SMC) to manage roll and yaw forces, while the attitude controller employs three types of hierarchical sliding mode control (HSMC) to ensure precise trajectory tracking of the UAVs. The main contributions of our proposed method compared to other related methods are outlined as follows:

\begin{itemize}
\item
We introduce an EKF for quadrotor UAVs, ensuring robust state estimation and enhancing controller efficacy. Compared to other nonlinear observers, the proposed EKF offers advantages in computational efficiency. By locally linearizing UAV dynamics around the working point, the EKF reduces computational demands, potentially lowering cost and size requirements for real-world UAV applications.
\item
We develop a novel hierarchical sliding mode control (HSMC) framework for quadrotor UAVs, investigating three distinct types: Aggregated HSMC (AHSMC), Incremental HSMC (IHSMC), and Combining HSMC (CHSMC). The stability of the quadrotor UAVs is assured through rigorous analysis based on Lyapunov stability principles.
\item
We provide simulations of the proposed method and comparisons to other related methods. The results demonstrate the effectiveness and feasibility of the control strategies.
\end{itemize}

The rest of the paper is presented as follows: Section II presents the Newton-Euler rigid body model of UAVs. The EKF is designed in Section III. Section IV consists of two parts: the first proposes three types of HSMC to control the UAVs, and the second analyzes the stability of the quadrotor UAVs based on Lyapunov stability principles. Furthermore, to demonstrate the feasibility and effectiveness of the proposed control, simulations and comparisons with existing methods are provided in Section V. The conclusion of the paper is presented in Section VI, and additional details can be found in the Appendix.

\section{NEWTON-EULER RIGID BODY MODEL}
The quadrotor UAVs are aligned with 4 propellers, each considering rotation velocity $\omega_i$, which creates forces and torques acting on the located center of mass (CoM). The model of the UAVs is shown in Figure~\ref{modelling}. 
To model the quadrotor UAVs, the following assumptions are made:
\begin{itemize}
    \item The UAV's structure is rigid and symmetrical.
    \item The angular motion of the UAVs is low amplitude.
    \item The propellers are supposed to be rigid; the thrust and drag of the propellers are proportional to the square of the propeller's speed.
\end{itemize}
\begin{figure}[!ht]
\centering
 {\resizebox*{8cm}{!}{\includegraphics{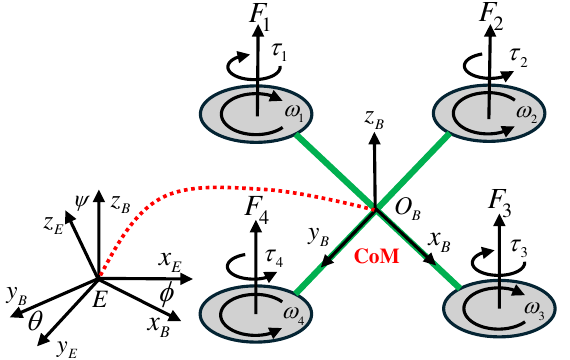}}}\hspace{5pt}
\caption{Quadrotor UAVs model.} 
\label{modelling}
\end{figure}
The linear velocity of the UAVs in the Earth frame is defined by: $\dot{\boldsymbol{\xi}} =[\dot{x},\dot{y},\dot{z}]^T$. The roll-pitch-yaw angle of the UAVs is defined by $\phi,\theta,\psi$. The Newton-Euler equation of the quadrotor UAVs motion is given as in \cite{bouadi2007modelling} by: 
\begin{align} \label{newdynamic}
\left\{\begin{array}{l}
m \ddot{\boldsymbol{\xi}}=\mathbf{F}_b+\mathbf{F}_{d}+\mathbf{F}_{g} \\
\mathbf{J}_M \dot{\boldsymbol{\Omega}}=-\boldsymbol{\Omega} \wedge \mathbf{J}_M \boldsymbol{\Omega}+\boldsymbol{\Gamma}_{f}-\boldsymbol{\Gamma}_{g}
\end{array}\right.
\end{align}
where $m$ is the mass of the UAVs; $\mathbf{F}_b, \mathbf{F}_{d}, \mathbf{F}_{g}$ are the resultant of the forces generated by the propeller rotors, the resultant of the drag forces, and the gravity forces, respectively. $\wedge$ denotes the cross product of two vectors. $\boldsymbol{\Omega}$ is the angular velocity in the body frame $O_B$ and is defined as:
\begin{align} \label{omegaframe}
\boldsymbol{\Omega}=\left(\begin{array}{ccc}
1 & 0 & -\sin \theta \\
0 & \cos \phi & \cos \theta \sin \phi \\
0 & -\sin \phi & \cos \phi \cos \theta
\end{array}\right)\left[\begin{array}{l}
\dot{\phi} \\
\dot{\theta} \\
\dot{\psi}
\end{array}\right].
\end{align}
In this research, we assume that the angular motion of the UAVs is low amplitude so that $\Omega$ can be assimilated to $[\dot{\phi}, \dot{\theta},  \dot{\psi}]^T$; and $\mathbf{J}_M$ is the inertial moments matrix. As mentioned before, the UAVs are symmetrical, and the inertial moment's matrix is defined as: 
\begin{align} \label{inertialmomentsmatrix}
    \mathbf{J}_M = \text{diag}(I_{xx},I_{yy},I_{zz}).
\end{align} Therefore, the element $-\boldsymbol{\Omega} \wedge \mathbf{J}_M \boldsymbol{\Omega}$ can be calculated as:
\begin{align} \label{elementtor}
    -\boldsymbol{\Omega} \wedge \mathbf{J}_M \boldsymbol{\Omega}  &= -\left[\begin{array}{l}
\dot{\phi} \\
\dot{\theta} \\
\dot{\psi}
\end{array}\right] \wedge  \left(\begin{array}{ccc}
I_{xx} & 0 & 0 \\
0 & I_{yy} & 0 \\
0 & 0 & I_{zz}
\end{array}\right) \left[\begin{array}{l}
\dot{\phi} \\
\dot{\theta} \\
\dot{\psi}
\end{array}\right]  \notag \\
&=  \left[\begin{array}{l}
\dot{\theta}\dot{\psi} (I_{yy}-I_{zz}) \\
\dot{\phi}\dot{\psi}(I_{zz}-I_{xx}) \\
\dot{\phi}\dot{\theta}(I_{xx}-I_{yy})
\end{array}\right].
\end{align}
The drag forces $\mathbf{F}_{d}$ is defined as:
\begin{align} \label{dragf}
\mathbf{F}_{d}=\left(\begin{array}{ccc}
-K_{dx} & 0 & 0 \\
0 & -K_{dy} & 0 \\
0 & 0 & -K_{dz}
\end{array}\right) \dot{\boldsymbol{\xi}}
\end{align}
The gravity forces: $\mathbf{F}_{g} = [0, 0 , -mg]^T $. The forces generated by the propeller rotor $\mathbf{F}_{b}$ can be calculated based on the total forces applied on the UAVs in the body frame by:
\begin{align} \label{ebforce}
\mathbf{F}_{b}=\mathbf{R}_{E}\left[\begin{array}{cc}
0  \\
0  \\
\sum_{i=1}^{4} F_i
\end{array}\right]
\end{align}
where $\mathbf{R}_E$ is the homogenous matrix transformation: 
\begin{align} \label{homogenousmatrix}
\mathbf{R}_{E} & =\left[\begin{array}{ccc}
c_{\psi} c_{\theta} & -s_{\psi} c_{\phi}+c_{\psi} s_{\theta} s_{\phi} & s_{\psi} s_{\phi}+c_{\psi} s_{\theta} c_{\phi} \\
s_{\psi} c_{\theta} & c_{\psi} c_{\phi}+s_{\psi} s_{\theta} s_{\phi} & -c_{\psi} s_{\phi}+s_{\psi} s_{\theta} c_{\phi} \\
-s_{\theta} & c_{\theta} s_{\phi} & c_{\theta} c_{\phi}
\end{array}\right] 
\end{align}
where $c_{\phi,\theta,\psi}=\cos{(\phi,\theta,\psi)}$, $s_{\phi,\theta,\psi}=\sin{(\phi,\theta,\psi)}$.\\
Moreover, $\boldsymbol{\Gamma}_{f}$ is the general moment generated by the propeller forces and torques: 
\begin{align} \label{generaltor}
\boldsymbol{\Gamma}_{f} & =\sum_{i=1}^{4} \boldsymbol{\tau}_{F_{i}}^{B}+\boldsymbol{\tau}_{{i}}^{B}=\sum_{i=1}^{4} \boldsymbol{p}_{O_{i}}^{B} \wedge \mathbf{F}_{i} + \sum_{i=1}^{4}{}^{B} \mathbf{R}_{O_{i}} \boldsymbol{\tau}_i =\left[\begin{array}{cc}
C_1 \\
C_2 \\
C_3
\end{array}\right]
\end{align}
where ${ }^{B} \mathbf{R}_{O_i}$ is the projection from propeller frame $O_{i}$ to $O_{B}$, $w_{i(i=1-4)}$ is the rotor velocity of the propeller $i$, and $\boldsymbol{p}_{O_{i}}^{B}$ is the vector of the body-frame to the origin of propeller at $O_i$. Substituting equation \eqref{elementtor},\eqref{dragf},\eqref{ebforce},\eqref{generaltor} to the Newtion-equation \eqref{newdynamic}, the full dynamic model of the UAVs is as follow:
\begin{align} \label{fullmodeluav}
m\ddot{x} & = (s_{\psi} s_{\phi}+c_{\psi} s_{\theta} c_{\phi}) F_z - K_{dx} \dot{x} \\ 
m\ddot{y} & =     (s_{\psi} s_{\theta} c_{\phi}-c_{\psi} s_{\phi}) F_z - K_{dy} \dot{y} \\
m\ddot{z} & =  c_{\theta} c_{\phi} F_z - K_{dz} \dot{z} -mg \\
I_{xx} \ddot{\phi} & = \dot{\theta}\dot{\psi} (I_{yy}-I_{zz}) + C_1 \\
I_{yy} \ddot{\theta} & =\dot{\psi}\dot{\phi}(I_{zz}-I_{xx}) + C_2 \\
I_{zz} \ddot{\psi} & = \dot{\phi}\dot{\theta}(I_{xx}-I_{yy}) + C_3 \label{fullmodeluavend}
\end{align}
where $F_z,  C_1,  C_2,  C_3$ are control inputs and are defined as:
\begin{align} \label{controlinput}
    \left[\begin{array}{c}
F_z \\
C_1 \\
C_2 \\
C_3
\end{array}\right] =  \left[\begin{array}{cccc}
k_t & k_t & k_t & k_t  \\
0 & -lk_t & 0 & lk_t  \\
-lk_t & 0 & lk_t & 0  \\
-k_d & k_d & -k_d & k_d  \\
\end{array}\right] \left[\begin{array}{c}
w_1^2 \\
w_2^2 \\
w_3^2 \\
w_4^2
\end{array}\right]
\end{align}
where $k_t$ and $k_d$ are the thrust and drag coefficients of the propellers, respectively, and $l$ is the length from the center of mass (CoM) to the propellers.
\section{EXTENDED KALMAN FILTER FOR QUADROTOR UAVS}
In this section, we present a strategy to estimate UAV's system states using EKF. The nonlinear model of the UAVs with the presence of noises can be written as:
 \begin{equation}
 \begin{cases}
 \dot{\mathbf{X}}(t)=\mathbf{F(X}(t),\mathbf{U}(t)) + \mathbf{W}(t) \\
 \mathbf{Y}(t)=\mathbf{C(X}(t))+ \mathbf{H}(t)
 \end{cases}
 \end{equation}
 where $\mathbf{X}(t)=[x,y,z,\phi,\theta,\psi,\dot{x},\dot{y}, \dot{z},\dot{\phi},\dot{\theta},\dot{\psi}]^T$ is the system state, $\mathbf{U}(t)=[F_x,C_1,C_2,C_3]^T$ is control input, $\mathbf{F(X}(t),\mathbf{U}(t)) $ denote the nonlinear system function which arecalculated from Eq.~\eqref{fullmodeluav} - Eq.~\eqref{fullmodeluavend}, $\mathbf{C(X}(t))=\mathbf{X}(t)$ is output function, and  $\mathbf{W}(t),\mathbf{H}(t)$ denote the total noises affected to the quadrotor UAVs. The noises functions $\mathbf{W}(t),\mathbf{H}(t)$ are considered as Gaussian function $N(0,R)$ and $N(0,Q)$, respectively.
 
 However, in order to apply this strategy, the UAVs will be linearized around the working point ${\mathbf{X}}_{k}$. The linear discrete time-varying dynamics of the UAVs is expressed as:
  \begin{equation}
 \begin{cases}
 \Delta{\mathbf{X}}_{k+1}=\mathbf{A}({\mathbf{X}}_{k},\mathbf{U}_{k}){\mathbf{X}}_{k}T_s + {\mathbf{B}}({\mathbf{X}}_{k},\mathbf{U}_{k}){\mathbf{U}}_{k}T_s+\mathbf{W}(k) \\
 \mathbf{Y}_{k+1}={\mathbf{X}}_{k+1}+ \mathbf{H}(k)
 \end{cases}
 \end{equation}
 where:
 \begin{align}
      \Delta{\mathbf{X}}_{k+1}&=\mathbf{X}_{k+1}-{\mathbf{X}}_{k}  \\
      \mathbf{A}({\mathbf{X}}_{k},\mathbf{U}_{k}) &= \frac{\partial\mathbf{F(X}(t),\mathbf{U}(t))}{\partial\mathbf{X}(t)}\vert_{\mathbf{X}_{k},\mathbf{U}_{k}} \label{Aline} \\
      \mathbf{B}({\mathbf{X}}_{k},\mathbf{U}_{k}) &= \frac{\partial\mathbf{F(X}(t),\mathbf{U}(t))}{\partial\mathbf{U}(t)}\vert_{\mathbf{X}_{k},\mathbf{U}_{k}} \label{Bline}.
 \end{align}
 $T_s$  is the sampling time, and $\mathbf{W}(k), \mathbf{H}(k)$ are noises at sample $k$. The details of matrix $\mathbf{A}({\mathbf{X}}_{k},\mathbf{U}_{k}), \mathbf{B}({\mathbf{X}}_{k},\mathbf{U}_{k}) $ can be seen in the Appendix.

 A state estimator based on EKF is presented as:
 \begin{align}
     \Delta\hat{\mathbf{X}}_{k+1}=&\mathbf{A}(\hat{\mathbf{X}}_{k},\mathbf{U}_{k})\hat{\mathbf{X}}_{k}T_s +\mathbf{B}(\hat{\mathbf{X}}_{k},\mathbf{U}_{k})\mathbf{U}_{k}T_s \notag \\
     &+ \mathbf{K}_{k+1}( \mathbf{Y}_{k+1}- \mathbf{C}_k\hat{\mathbf{X}}_{k})
 \end{align}
 and the filter gain is defined based on error covariance $\mathbf{P}_{k+1}$ as:
 \begin{align}
          \mathbf{K}_{k+1}= \mathbf{A}_k\mathbf{P}_{k+1}\mathbf{C}_k^T(\mathbf{C}_k\mathbf{P}_{k+1}\mathbf{C}_k^T+\mathbf{R}_k)^{-1}
 \end{align}
 where the covariance $\mathbf{P}_{k+1}$ is defined in each cycle by:
 \begin{align}
    \mathbf{P}_{k+1} = \mathbf{A}_{k} \mathbf{P}_k \mathbf{A}_k^T+ \mathbf{Q}_k- \mathbf{K}_k(\mathbf{C}_k\mathbf{P}_k\mathbf{C}_k^T+\mathbf{R}_k)\mathbf{K}_k^T
 \end{align}
where $\mathbf{A}_k=\mathbf{A}(\hat{\mathbf{X}}_{k},\mathbf{U}_{k})$, $\mathbf{C}_k$ is identity matrix due to $ \mathbf{Y}_{k+1}={\mathbf{X}}_{k+1}+ \mathbf{H}(k)$, and  $\mathbf{R}_k,\mathbf{Q}_k$ are the covariance noise matrices.

In this research, we consider the discrete-time EKF like in \cite{tonne2007stability,boutayeb1997convergence}, taking account of a priori and a posteriori states, the EKF becomes:\\
Linearization:
 \begin{align} \label{linearzing}
     \mathbf{A}(\hat{\mathbf{X}}_{k},\mathbf{U}_{k}) &= \frac{\partial\mathbf{F(X}(t),\mathbf{U}(t))}{\partial\mathbf{X}(t)}\vert_{\hat{\mathbf{X}}_{k},\mathbf{U}_{k}} \\
      \mathbf{B}(\hat{\mathbf{X}}_{k},\mathbf{U}_{k}) &= \frac{\partial\mathbf{F(X}(t),\mathbf{U}(t))}{\partial\mathbf{U}(t)}\vert_{\hat{\mathbf{X}}_{k},\mathbf{U}_{k}}
 \end{align}
Estimating: 
\begin{align}
      \Delta&\hat{\mathbf{X}}_{k+1}^-=\mathbf{A}(\hat{\mathbf{X}}_{k},\mathbf{U}_{k})\hat{\mathbf{X}}_{k}T_s +\mathbf{B}(\hat{\mathbf{X}}_{k},\mathbf{U}_{k})\mathbf{U}_{k}T_s \label{priorx} \\
      &\mathbf{P}_{k+1}^- = \mathbf{A}(\hat{\mathbf{X}}_{k},\mathbf{U}_{k}) \mathbf{P}_k \mathbf{A}^T(\hat{\mathbf{X}}_{k},\mathbf{U}_{k})+ \mathbf{Q}_k  \label{priorp}
 \end{align}
Updating:
   \begin{align}
      &\mathbf{K}_{k+1}= \mathbf{P}^{-} _{k+1}\mathbf{C}_k^T(\mathbf{C}_k\mathbf{P}^{-} _{k+1}\mathbf{C}_k^T+\mathbf{R}_k)^{-1} \label{kgain} \\
      &\mathbf{P}_{k+1} = (\mathbf{I}-\mathbf{K}_{k+1}\mathbf{C}_k)\mathbf{P}_{k+1}^-  \label{pafter} \\
      & \Delta\hat{\mathbf{X}}_{k+1}=\Delta \hat{\mathbf{X}}_{k+1}^- +\mathbf{K}_{k+1} (\mathbf{Y}_{k+1}- \mathbf{C}_k\hat{\mathbf{X}}_{k+1}^-) \notag \\
      & \ \ \ \ \ \ \  =\Delta \hat{\mathbf{X}}_{k+1}^- +\mathbf{K}_{k+1} (\mathbf{X}_{k+1}- \hat{\mathbf{X}}_{k+1}^-)  . \label{updatexk1}
 \end{align} 
The full diagram of EKF can be seen in Figure~\ref{EKFdiagram}.
\begin{figure}[!t]
\centering
{\resizebox*{7.5cm}{!}{\includegraphics{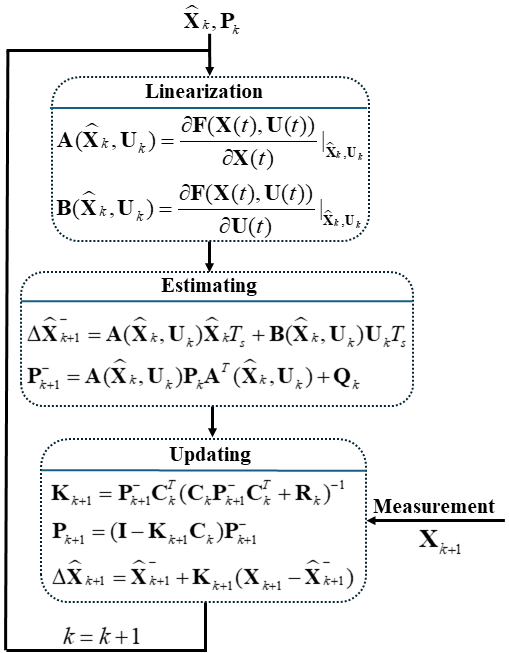}}}\hspace{5pt}
\caption{Block diagram of the Extended Kalman Filter.} 
\label{EKFdiagram}
\end{figure}
\begin{theorem} \label{theoremekf}
     Consider an EKF as in \eqref{linearzing}-\eqref{updatexk1}, and the following assumptions hold: \\
     1. There are positive numbers $a,c,p_1,p_2$ such that the following equations are satisfy with all $k \geq 0$:
     \begin{align}
        & \| \mathbf{A}_k \| \leq a \\
         & \| \mathbf{C}_k \| \leq c \\
         & p_1 \mathbf{I} \leq \| \mathbf{P}_k^- \| \leq p_2 \mathbf{I} \\
         &  p_1 \mathbf{I} \leq \| \mathbf{P}_k^+ \| \leq p_2 \mathbf{I}. 
     \end{align}
     2. $\mathbf{A}_k $ is nonsingular with all $k\geq 0$\\
     then the EKF is an exponential observer, and the observer error $\epsilon=\mathbf{X}_k-\hat{\mathbf{X}}_{k}$ will be bounded.
\end{theorem}
\begin{proof}
Following the definition of $\mathbf{A}_k$ in the Appendix, and due to $\mathbf{C}_k= \mathbf{I}$, the assumption above will be held. The proof of the EKF's stability can be found in \cite{tonne2007stability,reif1999extended}.
\end{proof}

\section{HIERARCHICAL SLIDING MODE CONTROL BASED ON EXTENDED STATE OBSERVER}
The hierarchical method is pointed out in \cite{wang2004design} for a class of second-order underactuated systems. In order to apply the HSMC for the UAVs, the model of the quadrotor UAVs in \eqref{fullmodeluav}-\eqref{fullmodeluavend} can be rewritten that:
\begin{equation}
 \begin{cases} \label{redyna}
 \dot{x}_1=x_2 \\
 \dot{x}_2=(s_{x_{11}} s_{x_{7}}+c_{x_{11}} s_{x_{9}} c_{x_{7}}) F_z - K_{dx}{x}_2 \\
 \dot{x}_3=x_4 \\
 \dot{x}_4=(s_{x_{11}} s_{x_{9}} c_{x_{7}}-c_{x_{11}} s_{x_{7}}) F_z - K_{dy}{x}_4 \\
 \dot{x}_5=x_6 \\
 \dot{x}_6=c_{x_{9}} c_{x_{7}} F_z - K_{dz}x_6 -mg \\
 \dot{x}_7=x_8 \\
 \dot{x}_8= x_{10}x_{12} (I_{yy}-I_{zz}) + C_1 \\
 \dot{x}_9=x_{10} \\
 \dot{x}_{10}=x_{8}x_{12}(I_{zz}-I_{xx}) + C_2 \\
 \dot{x}_{11}=x_{12} \\
 \dot{x}_{12}=x_{8}x_{10}(I_{xx}-I_{yy}) + C_3 \\
 \end{cases}
 \end{equation}
where $x_1=x$, $x_2=\dot{x}$, $x_3=y$, $x_4=\dot{y}$ \, $x_5=z$, $x_6=\dot{z}$, $x_7=\phi$, $x_8=\dot{\phi}$, $x_9=\theta$, $x_{10}=\dot{\theta}$, $x_{11}=\psi$,$x_{12}=\dot{\psi}$, $s_{x_i}=sin(x_i)$, and $c_{x_i}=cos(x_i)$.
Or in this form:
\begin{equation} \label{statedyna}
 \begin{cases}
 \dot{x}_1=x_2 \\
 \dot{x}_2=f_x+b_xu_1 \\
 \dot{x}_3=x_4 \\
 \dot{x}_4=f_y+b_yu_1 \\
 \dot{x}_5=x_6 \\
 \dot{x}_6=f_z+bzu_1 \\
 \dot{x}_7=x_8 \\
 \dot{x}_8=f_{\phi}+b_{\phi}u_2 \\
 \dot{x}_9=x_{10} \\
 \dot{x}_{10}=f_{\theta}+b_{\theta}u_3\\
 \dot{x}_{11}=x_{12} \\
 \dot{x}_{12}=f_{\psi}+b_{\psi}u_4\\
 \end{cases}
 \end{equation}
where $[u_1,u_2,u_3,u_4]=[F_z,C_1,C_2,C_3]$, and $f_{i},b_{i}(i=x,y,z,\phi,\theta,\psi)$ can be identified based on \eqref{redyna}, respectively.
\begin{assumption} \label{assumerror}
Based on Theorem~\ref{theoremekf}, it is assumed that during the controller design process for the quadrotor UAVs, the estimation error will be bounded and converged to zero. Consequently, the sliding surface and the Lyapunov function will disregard these values. The controller design and its stability are structured like state feedback.
\end{assumption}
\subsection{PD-SMC attitude controller for UAVs}
In this section, the PD-SMC attitude controller (PD-SMC-AC) is proposed to create the desired roll and pitch angles for the UAVs. For the attitude controller,  we consider  $\phi,\theta,\psi$ (the roll, pitch, yaw angles) of the quadrotor UAVs as the class of actuated states, which can be controlled by $u_2, u_3, u_4$, respectively.

Considering the sliding surface as:
\begin{align}
    s_\phi= c_{\phi}e_7+e_8 
\end{align}
where $c_{\phi}$ is a positive constant, and $e_7=\hat{\phi}-\phi_r,e_8=\dot{\hat{\phi}}-\dot{\phi}_r$, and $\phi_r$ is the desired roll angle.\\
The Lyapunov function is chosen as:
\begin{align}
    V_{\phi}=\frac{1}{2} s_\phi^2 
\end{align}
and its derivative:
\begin{align} \label{lyasphi}
    \dot{V}_{\phi}&=s_\phi\dot{s}_\phi \notag \\
    &=s_\phi (c_{\phi}(\dot{\hat{\phi}}-\dot{\phi}_r)+ \ddot{\phi}-\ddot{\phi}_r-\ddot{\epsilon}_{\phi}) \notag \\
    &\approx s_\phi (c_{\phi}(\dot{\hat{\phi}}-\dot{\phi}_r)+ \ddot{\phi}-\ddot{\phi}_r)
\end{align}
where $\epsilon_{\phi}$ is the observed error of roll angle. Therefore, the control law $u_2$ is proposed as:
\begin{align} \label{controlphi}
    u_2&=\frac{1}{b_\phi}[-c_{\phi}\hat{\phi}-(c_{\phi}+1)\dot{\hat{\phi}}+c_{\phi}\phi_r +(c_{\phi}+1)\dot{\phi}_r  \notag \\
     & \ \ \ \ \ \ -f_{\phi}-K_{\phi}\text{sat}(s_\phi)+\ddot{\phi}_r]
\end{align}
where $K_{\phi}$ is a positive number, $\text{sat}(s_\phi)$ is the saturation sign function, which is defined by:
\begin{align} \label{satsign}
    \text{sat}(s_\phi)=
 \begin{cases} 
  -1 \text{ if } s_\phi \leq -1 \\
  s_\phi \text{ if } -1 < s_\phi < 1 \\
  1 \text{ if } s_\phi \geq 1. \\
 \end{cases}
\end{align}
Substituting the control law \eqref{controlphi}, and the dynamic equation \eqref{statedyna} into \eqref{lyasphi} with the note that the estimation error will converge to 0, the derivative of the Lyapunov candidate becomes:
\begin{align}
    \dot{V}_{\phi}&=s_\phi [-c_{\phi}(\hat{{\phi}}-{\phi}_r)- (\dot{\hat{\phi}}-\dot{\phi}_r) -K_{\phi}\text{sat}(s_\phi)] \notag \\
    & = -s_\phi^2-s_{\phi} K_{\phi}\text{sat}(s_\phi) .
\end{align}
Hence, the system is stable. Moreover, applying Barbalat’s lemma \cite{slotine1991applied}, the sliding surface will converge to zero, which means  $c_{\phi}e_7+e_8  \rightarrow 0$. Solving the differential equation we have the control error ${\phi}-\phi_r \rightarrow 0$, and then the UAVs can track references.
For the pitch and yaw angles, let's define the sliding surface as:
\begin{align}
    & s_\theta= c_{\theta}e_9+e_{10}  \\
    & s_\psi= c_{\psi}e_{11}+e_{12} .
\end{align}
Defining the Lyapunov function as: 
\begin{align}
    & V_{\theta}=\frac{1}{2} s_\theta^2  \\
    & V_{\psi}=\frac{1}{2} s_\psi^2.
\end{align}
Using the same technique with the roll angle, the control input $u_3,u_4$ can be obtained by:
\begin{align} \label{controltheta}
u_3&=\frac{1}{b_\theta}[-c_{\theta}\hat{\theta}-(c_{\theta}+1)\dot{\hat{\theta}}+c_{\theta}\theta_r +(c_{\theta}+1)\dot{\theta}_r  \notag \\
     & \ \ \ \ \ \ -f_{\theta}-K_{\theta}\text{sat}(s_\theta)+\ddot{\theta}_r] \\
u_4&=\frac{1}{b_\psi}[-c_{\psi}\hat{\psi}-(c_{\psi}+1)\dot{\hat{\psi}}+c_{\psi}\psi_r +(c_{\psi}+1)\dot{\psi}_r  \notag \\
     & \ \ \ \ \ \ -f_{\psi}-K_{\psi}\text{sat}(s_\psi)+\ddot{\psi}_r]   \label{controlpsi}
\end{align}
then the SMC can help the UAVs track the roll-pitch-yaw references.\\
However, for the reference roll and pitch angle, we consider a group of PD controllers for the tracking along the $x$-axis and the $y$-axis. To move along the $x$-direction the quadrotor UAVs are required to roll around the $y$-axis and to move along the $y$-direction the quadrotor UAVs are required to roll around the $x$-axis.
The roll angle is approximated as the first-order system of $(y_r – y)$, where $y_r$ is the desired trajectory.
Therefore, the PD controller to create the roll and the yaw reference is defined as: \\
Roll controller:
\begin{align}
    \phi_r=k_{py}(y_r-\hat{y})-k_{dy}(\dot{y}_r-\dot{\hat{y}})
\end{align}
Pitch controller:
\begin{align}
    \theta_r=k_{px}(\hat{x}-x_r)-k_{dx}(\dot{\hat{x}}-\dot{x}_r)
\end{align}
where $k_{px},k_{dx},k_{py},k_{dy}$ are the controller parameters, respectively. 
%The movement of the quadrotor UAVs can be seen in \cite{milhim2010gain}. 
%\begin{figure}[!ht]
%\centering
%{\resizebox*{7.75cm}{!}{\includegraphics{figures/movementuavs.pdf}}}\hspace{5pt}
%\caption{Roll and Pitch movement of the UAVs.} 
%\label{movementuav} 
%\end{figure}
\subsection{HSMC position controller for UAVs}
In this part, we propose three types of HSMC: Aggregated HSMC (AHSMC), Incremental HSMC (IHSMC), and Combining HSMC (CHSMC). Inspired by the idea in \cite{qian2016hierarchical}, the HSMC is designed for three class variables $x,y,z$ with one control input $F_z$. Its state space expression is represented by:
\begin{align}
    \begin{cases}
        \dot{x}_{2i-1}=x_{2i} \\
        \dot{x}_{2i}=f_i+b_iu
    \end{cases}
\end{align}
where $i=1,2,3$, and $u=F_z$.
\subsubsection{Aggregated HSMC}
\begin{figure}[!ht]
\centering
\begin{tikzpicture}
\filldraw[black, very thick] (0,0) rectangle (1.5,0.1);
\filldraw[black, very thick] (2,0) rectangle (3.5,0.1);
\draw[->] (0.25,-0.95) node[below]{$e_1$} -- (0.25,-0.05);
\draw[->] (1.25,-0.95) node[below]{$e_2$} -- (1.25,-0.05);
\draw[->] (2.25,-0.95) node[below]{$e_3$} -- (2.25,-0.05);
\draw[->] (3.25,-0.95) node[below]{$e_4$} -- (3.25,-0.05);
\draw[->] (0.75,0.05) -- (0.75,1.05);
\draw (0.75,0.55) node[anchor=west]{$S_1$};
\draw[->] (2.75,0.05) -- (2.75,1.05);
\draw (2.75,0.55) node[anchor=west]{$s_2$};
\filldraw[black, very thick] (0.5,1.1) rectangle (3,1.2);
\draw[->] (1.75,1.15) -- (1.75,2.15);
\draw (1.75,1.65) node[anchor=west]{$S_2$};
\filldraw[black, very thick] (4,1.1) rectangle (5.5,1.2);
\draw[->] (4.25,0.1) node[below]{$e_5$} -- (4.25,1.05);
\draw[->] (5.25,0.1) node[below]{$e_6$} -- (5.25,1.05);
\draw[->] (4.75,1.15) -- (4.75,2.15);
\draw (4.75,1.65) node[anchor=west]{$s_3$};
\filldraw[black, very thick] (1.5,2.2) rectangle (5,2.3);
\draw[->] (3.25,2.3) -- (3.25,3.3)  node[above]{$S_3$} ;
\end{tikzpicture}
\caption{Structure of AHSMC.} 
\label{ahsmcstructure}
\end{figure}
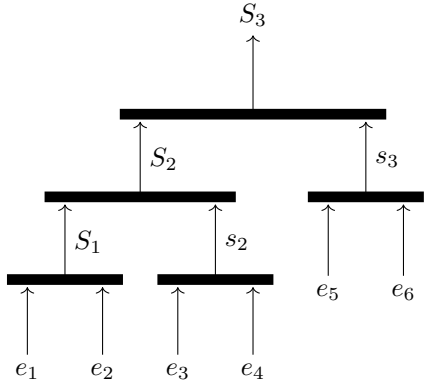
The basic idea of AHSMC is to pair two variable states as a lower layer, and the higher layer considering the lower layers. The structure of AHSMC is depicted in Figure~\ref{ahsmcstructure}.
The lower layer is defined as follows:
\begin{align}
    &s_1=c_1e_1+e_2 \\
    &s_2=c_2e_3+e_4 \\
    &s_3=c_3e_5+e_6 
\end{align}
where $c_1,c_2,c_3$ are positive numbers, $e_1=\hat{x}-x_r$, $e_2=\dot{\hat{x}}-\dot{x}_r$, $e_3=\hat{y}-y_r$, $e_4=\dot{\hat{y}}-\dot{y}_r$, $e_5=\hat{z}-z_r$, $e_6=\dot{\hat{z}}-\dot{z}_r$ are control errors. \\
The higher layers are defined as:
\begin{align}
    &S_1=s_1 \\
    &S_2=\lambda_1S_1+ s_2 \\
    &S_3=\lambda_2S_2+s_3
\end{align}
where $\lambda_1,\lambda_2$ are control parameter. The equivalent controller that makes the sliding surface's derivative equal to zero is calculated based on:
\begin{align} \label{deris1}
    \dot{s}_1&=c_1\dot{e}_1+\dot{e}_2 = c_1(\dot{\hat{x}}-\dot{x}_r)+(\ddot{\hat{x}}-\ddot{{x}}_r) \notag \\
    &=  c_1(\dot{\hat{x}}-\dot{x}_r)+(f_{\hat{x}}+b_{\hat{x}}u-\ddot{x}_r).
\end{align}
Therefore the equivalent control law for sliding surface $s_1$ is chosen as:
\begin{align}
    u_{eqx}=\frac{-c_1(\dot{\hat{x}}-\dot{x}_r)+\ddot{x}_r-f_{\hat{x}}}{b_{\hat{x}}}.
\end{align}
Similar to the equivalent control law of sliding surface $s_2,s_3$:
\begin{align}
    u_{eqy}&=\frac{-c_2(\dot{\hat{y}}-\dot{y}_r)+\ddot{y}_r-f_{\hat{y}}}{b_{\hat{y}}}\\
     u_{eqz}&=\frac{-c_3(\dot{\hat{z}}-\dot{z}_r)+\ddot{z}_r-f_{\hat{z}}}{b_{\hat{z}}}.
\end{align}
The switching control law of the AHSMC is proposed as:
\begin{align}
    u_{sw}=&-\frac{\lambda_1\lambda_2b_{\hat{x}}(u_{eqy}+u_{eqz})+\lambda_2b_{\hat{y}}(u_{eqx}+u_{eqz})  }{\lambda_2\lambda_1b_{\hat{x}}+\lambda_2b_{\hat{y}}+b_{\hat{z}}}\notag \\
    &- \frac{b_{\hat{z}}(u_{eqx}+u_{eqy})+K_a S_3+\eta \text{sat}(S_3)}{\lambda_2\lambda_1b_{\hat{x}}+\lambda_2b_{\hat{y}}+b_{\hat{z}}}
\end{align}
where $K_a,\eta$ are positive numbers, and $\text{sat}(S_3)$ is the saturation sign function, which can be defined as in \eqref{satsign}.
\begin{theorem}
    In the view of Assumption~\ref{assumerror}, considering the quadrotor UAVs with the state space \eqref{statedyna}, if the AHSMC controller is proposed as:
    \begin{align} \label{controlahsmc}
        u= u_{eqx}+u_{eqy}+u_{eqz}+u_{sw}
    \end{align}
    then the sliding surface $S_3$  is asymptotically stable.
\end{theorem}
\begin{proof}
    According to Lyapunov's stability, considering the Lyapunov candidate as:
    \begin{align}
        V_{a}=\frac{1}{2}S_3^2.
    \end{align}
     Differentiating $V_{a}$ with respect to time $t$, the derivative of $V_{a}$ yield is expressed as:
     \begin{align} \label{dotva}
         \dot{V}_a=S_3\dot{S}_3=S_3(\lambda_1\lambda_2\dot{s}_1+\lambda_2\dot{s}_2+\dot{s}_3).
     \end{align}
     Substituting the control law \eqref{controlahsmc} ,\eqref{deris1} and the similar $\dot{s}_2,\dot{s}_3$ into the state space \eqref{dotva}, the derivative becomes:
     \begin{align}
         \label{redotva}
          \dot{V}_a=&S_3 \bigg{ \{ }\lambda_1\lambda_2 [ c_1(\dot{\hat{x}}-\dot{x}_r)+(f_{\hat{x}}+b_{\hat{x}}u-\ddot{x}_r)] \notag \\ 
          & \ \ \ \ \ \ + \lambda_2 [ c_2(\dot{\hat{y}}-\dot{y}_r)+(f_{\hat{y}}+b_{\hat{y}}u-\ddot{y}_r)] \notag \\
          & \ \ \  \ \ \ + c_3(\dot{\hat{z}}-\dot{z}_r)+(f_{\hat{z}}+b_{\hat{z}}u-\ddot{z}_r)\bigg{ \} } \notag \\
          =& S_3 \bigg{ [ }\lambda_1\lambda_2 b_{\hat{x}}(u_{eqy}+u_{eqz}+u_{sw}) + \lambda_2b_{\hat{y}} (u_{eqx} \notag \\ & \ \ \ \ \  +u_{eqz} +u_{sw})  + b_{\hat{z}}(u_{eqx}+u_{eqy}+u_{sw}) \bigg{ ] } \notag \\
          =& -K_aS_3^2-\eta S_3 \text{sat}(S_3).
     \end{align}
\end{proof}
Hence, the sliding surface is stable. Moreover, applying Barbalat’s lemma \cite{slotine1991applied}, the sliding surface will converge to zero. Since the sliding surface $S_3$ is asymptotically stable from the time domain $[0,t_f]$, the stability of sliding surface $S_1, S_2$  can be archived in the time domain $[t_f,\infty]$. The proof of this stability is in \cite{qian2016hierarchical}.
\subsubsection{Incremental HSMC}
\begin{figure}[!ht]
\centering
\begin{tikzpicture}
\draw[->] (0.25,-0.95) node[below]{$e_1$} -- (0.25,-0.05);
\draw[->] (1.25,-0.95) node[below]{$e_2$} -- (1.25,-0.05);
\draw[->] (2.25,-0.95) node[below]{$e_3$} -- (2.25,1);
\draw[->] (3.25,-0.95) node[below]{$e_4$} -- (3.25,2.05);
\draw[->] (4.25,-0.95) node[below]{$e_5$} -- (4.25,3.1);
\draw[->] (5.25,-0.95) node[below]{$e_6$} -- (5.25,4.15);
\filldraw[black, very thick] (0,0) rectangle (1.5,0.1);
\draw[->] (0.75,0.1) -- (0.75,1.0);
\draw (0.75,0.55) node[anchor=east]{$s_1$};
\filldraw[black, very thick]  (0.5,1.05)  rectangle (2.5,1.15);
\draw[->] (1.5,1.15) -- (1.5,2.05);
\draw (1.55,1.6) node[anchor=east]{$s_2$};
\filldraw[black, very thick]  (1.25,2.1)  rectangle (3.5,2.2);
\draw[->] (2.25,2.2) -- (2.25,3.1);
\draw (2.25,2.65) node[anchor=east]{$s_3$};
\filldraw[black, very thick]  (2,3.15)  rectangle (4.5,3.25);
\draw[->] (3.25,3.25) -- (3.25,4.15);
\draw (3.25,3.75) node[anchor=east]{$s_4$};
\filldraw[black, very thick]  (3,4.2)  rectangle (5.5,4.3);
\draw[->] (4.25,4.35) -- (4.25,5.25);
\draw (4.25,5.65) node[anchor=north]{$s_5$};
\end{tikzpicture}
\caption{Structure of IHSMC.} 
\label{ihsmcstructure}
\end{figure}
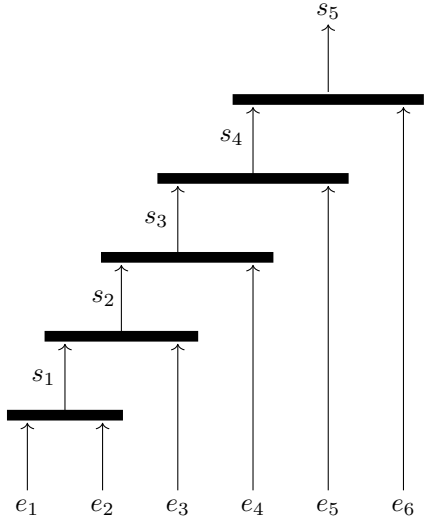
The idea of IHSMC is to add an additional state into each layer. Without the loss of generality, the first layer is the construction of control error along $x_1,x_2$ variables, and the higher layer takes control error of the next variable state. As a result, the structure of IHSMC is depicted in Figure~\ref{ihsmcstructure}.
The sliding surface of the IHSMC is designed as:
\begin{align}
    s_1=c_1e_2+e_1  \\
    s_2=c_2e_3+s_1 \\
    s_3=c_3e_4+s_2  \\
    s_4=c_4e_5+s_3  \\
    s_5=c_5e_6+s_4 
\end{align}
where $c_1, c_2, c_3, c_4, c_5$ are positive numbers, and the error $e_{i,i=1-6}$ is defined as the same of the AHSMC.
\begin{theorem}
    In the view of Assumption~\ref{assumerror}, considering the quadrotor UAVs with the state space \eqref{statedyna}, if the IHSMC controller is proposed as:
    \begin{align} \label{controlihsmc}
        u= u_{eq}+u_{sw}
    \end{align}
where $u_{eq},u_{sw}$ are the equivalent control law and switch control law, respectively, then the sliding surface $s_5$  is asymptotically stable.
\end{theorem}
\begin{proof}
According to Lyapunov's stability, considering the Lyapunov candidate as:
    \begin{align}
        V_{i}=\frac{1}{2}s_5^2
    \end{align}
Differentiating $V_{i}$ with respect to time $t$:
\begin{align} \label{dotvi}
        \dot{V}_i&=s_5\dot{s}_5=s_5(c_5\dot{e}_6+c_4\dot{e}_5+c_3\dot{e}_4+c_2\dot{e}_3+c_1\dot{e}_2+\dot{e}_1) \notag \\
         &=s_5\big{ [ } c_5(\ddot{\hat{z}}-\ddot{z}_r) + c_4(\dot{\hat{z}}-\dot{z}_r) +c_3(\ddot{\hat{y}}-\ddot{y}_r) \notag \\
         &\ \ \ \ +c_2(\dot{\hat{y}}-\dot{y}_r) +c_1(\ddot{\hat{x}}-\ddot{x}_r) +(\dot{\hat{x}}-\dot{x}_r) \big{ ] }  \notag \\
         &= s_5\big{ [ } c_5(f_{\hat{z}}+b_{\hat{z}}u-\ddot{z}_r) + c_3(f_{\hat{y}}+b_{\hat{y}}u-\ddot{y}_r) \notag \\
         & \ \ \ \ + c_1(f_{\hat{x}}+b_{\hat{x}}u-\ddot{x}_r) + c_4(\dot{\hat{z}}-\dot{z}_r) \notag \\
         & \ \ \ \ +c_2(\dot{\hat{y}}-\dot{y}_r) + (\dot{\hat{x}}-\dot{x}_r) \big{ ] }.
     \end{align}
     Therefore, the equivalent control law and the switch control law are proposed as:
     \begin{align}
         u_{eq}&= -\frac{c_5f_{\hat{z}}+ c_3f_{\hat{y}}+c_1f_{\hat{x}}+ c_4(\dot{\hat{z}}-\dot{z}_r)+ c_2(\dot{\hat{y}}-\dot{y}_r)}{c_5b_{\hat{z}}+c_3b_{\hat{y}}+c_1b_{\hat{x}}} \notag \\
         & \ \ \ - \frac{ (\dot{\hat{x}}-\dot{x}_r) - c_5\ddot{z}_r-c_3\ddot{y}_r-c_1\ddot{x}_r }{c_5b_{\hat{z}}+c_3b_{\hat{y}}+c_1b_{\hat{x}}} \\
          u_{sw}&= -\frac{K_is_5+\eta \text{sat}(s_5)}{c_5b_{\hat{z}}+c_3b_{\hat{y}}+c_1b_{\hat{x}}}
     \end{align}
     with $K_i$ is a positive number, and $\text{sat}(s_5)$ is the saturation sign function.
     Substituting the control law \eqref{controlihsmc} into the \eqref{dotvi}, the derivative becomes:
     \begin{align}
         \label{redotvi}
          \dot{V}_i= -K_is_5^2-\eta s_5 \text{sat}(s_5).
     \end{align}
     Hence, the sliding surface is stable. To check the stability of all the element sliding surfaces $s_i$, the proof is provided in \cite{qian2016hierarchical}.
\end{proof}

\subsubsection{Combining HSMC}
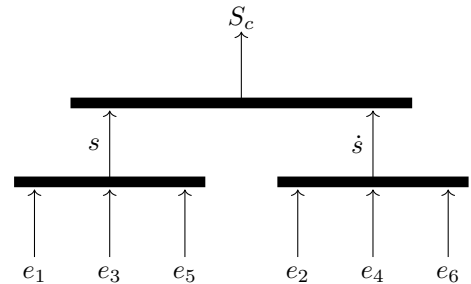
\begin{figure}[!ht]
\centering
\begin{tikzpicture}
\draw[->] (0.25,-0.95) node[below]{$e_1$} -- (0.25,-0.05);
\draw[->] (1.25,-0.95) node[below]{$e_3$} -- (1.25,-0.05);
\draw[->] (2.25,-0.95) node[below]{$e_5$} -- (2.25,-0.05);
\draw[->] (3.75,-0.95) node[below]{$e_2$} -- (3.75,-0.05);
\draw[->] (4.75,-0.95) node[below]{$e_4$} -- (4.75,-0.05);
\draw[->] (5.75,-0.95) node[below]{$e_6$} -- (5.75,-0.05);
\filldraw[black, very thick] (0,0) rectangle (2.5,0.1);
\draw[->] (1.25,0.1) -- (1.25,1);
\draw (1.25,0.55) node[anchor=east]{$s$};
\filldraw[black, very thick] (3.5,0) rectangle (6,0.1);
\draw[->] (4.75,0.1) -- (4.75,1);
\draw (4.75,0.55) node[anchor=east]{$\dot{s}$};
\filldraw[black, very thick] (0.75,1.05) rectangle (5.25,1.15);
\draw[->] (3,1.15) -- (3,2.05);
\draw (3,2.5) node[anchor=north]{$S_c$};
\end{tikzpicture}
\caption{Structure of CHSMC.} 
\label{chsmcstructure}
\end{figure}

\begin{figure*}
\centering
{\resizebox*{16cm}{!}{\includegraphics{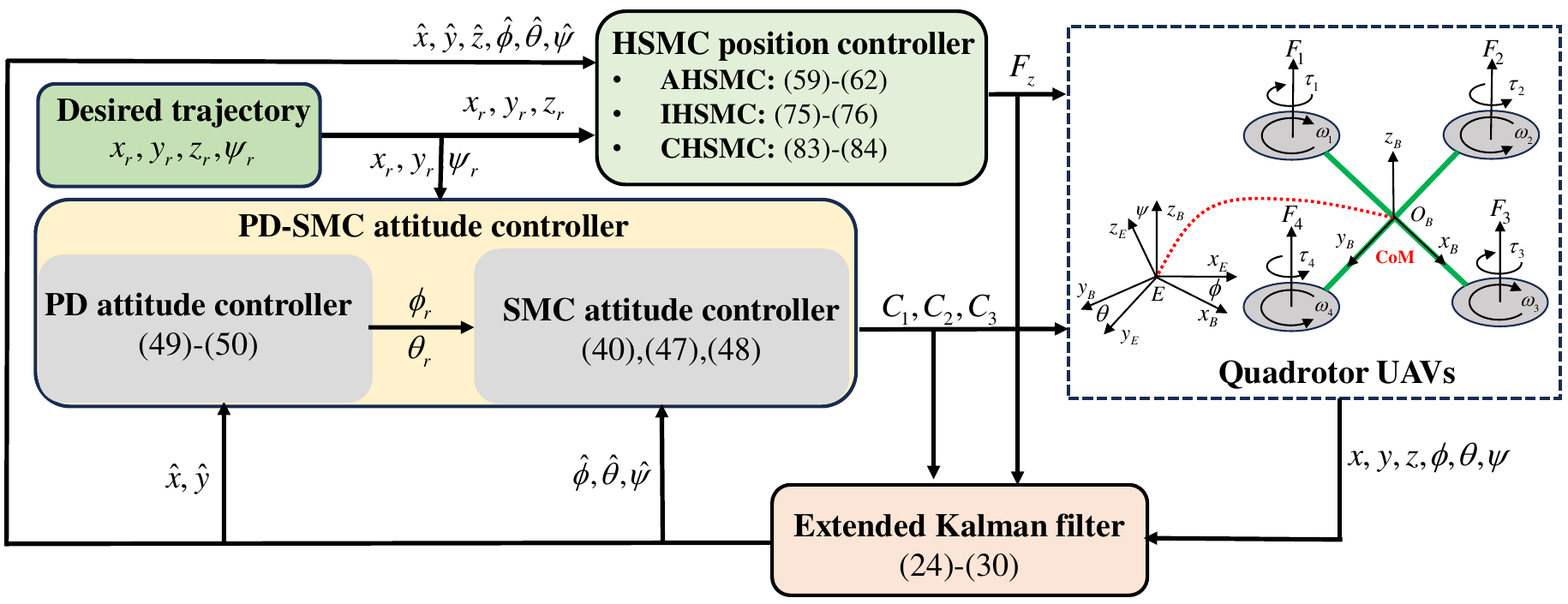}}}\hspace{5pt}
\caption{Close-loop control structure of the proposed method.} 
\label{schemeuavs}
\end{figure*}

The idea of the CHSMC is to divide the system state into two parts: the derivative combination and the linear combination. The structure of CHSMC is depicted in Figure~\ref{chsmcstructure}, where $s$ represents the linear combination, and $\dot{s}$ represents the derivative combination.
The element sliding surface of the CHSMC is proposed as:
\begin{align} \label{chsmcslideele}
    s=c_1e_1+c_2e_3+c_3e_5 \notag \\
    \dot{s}=c_1e_2+c_2e_4+c_3e_6
\end{align}
where $c_1,c_2,c_3$ are control parameters, and the control error is defined as the same as the AHSMC and the IHSMC.
The CHSMC sliding surface is proposed as:
\begin{align} \label{shsmcslide}
    S_c= \alpha s + \dot{s}.
\end{align}
\begin{theorem}
    In the view of Assumption~\ref{assumerror}, considering the quadrotor UAVs with the state space \eqref{statedyna}, if the CHSMC controller is proposed as:
    \begin{align} \label{controlchsmc}
        u= u_{eq}+u_{sw}
    \end{align}
where $u_{eq},u_{sw}$ are the equivalent control law and switch control law respectively, then the sliding surface $S_c$  is asymptotically stable.
\end{theorem}
\begin{proof}
According to Lyapunov's stability, considering the Lyapunov candidate as:
    \begin{align}
        V_{c}=\frac{1}{2}S_c^2.
    \end{align}
Differentiating $V_{c}$ with respect to time $t$:
\begin{align} \label{dotvc}
        \dot{V}_c&=S_c\dot{S}_c=S_c(\alpha \dot{s} +\ddot{s}) \notag \\
         &=S_c\big{ [ } \alpha c_1(\dot{\hat{x}}-\dot{x}_r) + \alpha c_2(\dot{\hat{y}}-\dot{y}_r) + \alpha c_3(\dot{\hat{z}}-\dot{z}_r)  \notag \\
         &\ \ \ \ \ \  +c_1(\ddot{\hat{x}}-\ddot{x}_r) + c_2(\ddot{\hat{y}}-\ddot{y}_r) + c_3(\ddot{\hat{z}}-\ddot{z}_r)  \big{ ] }  \notag \\
         &= S_c \big{ [ } \alpha c_1(\dot{\hat{x}}-\dot{x}_r) + \alpha c_2(\dot{\hat{y}}-\dot{y}_r) + \alpha c_3(\dot{\hat{z}}-\dot{z}_r)  \notag \\
         & \ \ \ \ \ \ + c_1(f_{\hat{x}}+b_{\hat{x}}u-\ddot{x}_r) +  c_2(f_{\hat{y}}+b_{\hat{y}}u-\ddot{y}_r) \notag \\
         & \ \ \ \ \ \ + c_3(f_{\hat{z}}+b_{\hat{z}}u-\ddot{z}_r) \big{ ] }.
     \end{align}
     Therefore, the equivalent control law and the switch control law are proposed as:
     \begin{align}
         u_{eq}&= -\frac{c_1f_{\hat{x}}+ c_2f_{\hat{y}}+c_3f_{\hat{z}}-c_1\ddot{x}_r -c_2\ddot{y}_r-c_3\ddot{z}_r}{c_1b_{\hat{x}}+c_2b_{\hat{y}}+c_3b_{\hat{z}}} \notag \\
         & \ \ \ - \frac{  \alpha c_1(\dot{\hat{x}}-\dot{x}_r) + \alpha c_2(\dot{\hat{y}}-\dot{y}_r) + \alpha c_3(\dot{\hat{z}}-\dot{z}_r) }{c_1b_{\hat{x}}+c_2b_{\hat{y}}+c_3b_{\hat{z}}} \\
          u_{sw}&= -\frac{K_c S_c+\eta \text{sat}(S_c)}{c_1b_{\hat{x}}+c_2b_{\hat{y}}+c_3b_{\hat{z}}}
     \end{align}
     with $K_c$ is a positive number, and $\text{sat}(S_c)$ is the saturation sign function.
     Substituting the control law \eqref{controlchsmc} into the \eqref{dotvc}, the derivative becomes:
     \begin{align}
         \label{redotvc}
          \dot{V}_c= -K_c S_c^2-\eta S_c \text{sat}(S_c).
     \end{align}
     Hence, the sliding surface is stable. To check the stability of all the element sliding surfaces $s$ and $ \dot{s}$, the proof is provided in \cite{qian2016hierarchical}.
\end{proof}
The closed-loop scheme for the proposed strategy is provided in Figure~\ref{schemeuavs}. 
\begin{algorithm}
\caption{HSMC based EKF strategy}\label{algorithm1}
1. At time $t=t_k$, measure the current state $x,y,z,\phi,\theta,\psi$. \\
2. Input the measured states and the previous control input into the EKF to estimate all the states for the controller: \\
2.1. Linearization of the system around the previous working states: \eqref{Aline},\eqref{Bline}.    \\
2.2. Estimating the prior estimate state and covariance: \eqref{priorx},\eqref{priorp}.  \\
2.3. Updating the posterior estimate, observer gain, and covariance: \eqref{kgain},\eqref{pafter},\eqref{updatexk1}  \\
3. Input the estimated states and the desired trajectory to the controller to calculate the control signals $F_z, C_1, C_2, C_3$ for the quadrotor UAVs.\\
4. Perform the quadrotor UAVs with these control laws. \\
5. Set $t=t_k+ T_s$ (with $T_s$ is the sample time) and go back to step 1.
\end{algorithm}

\section{Experimental Results}
In this section, to prove the effectiveness and performance of the proposed HSMC-EKF strategy, a comparison to the other methods PID-EKF and SO-SMC-EKF is provided. The control law of the PID-EKF is proposed based on \cite{salih2010modelling} as:
\begin{align}
    F_z &= k_{pz} (z_r-\hat{z}) + k_{dz}(\dot{z}_r-\dot{\hat{z}}) +mg  \\
    C_1 &=  k_{p\phi} (\phi_r-\hat{\phi}) + k_{d\phi}(\dot{\phi}_r-\dot{\hat{\phi}})\\
    C_2 &= k_{p\theta} (\theta_r-\hat{\theta}) + k_{d\theta}(\dot{\theta}_r-\dot{\hat{\theta}})\\
    C_3 &= k_{p\psi} (\psi_r-\hat{\psi}) + k_{d\psi}(\dot{\psi}_r-\dot{\hat{\psi}}).
\end{align}
In the PID-EKF, we choose the integral gain as zero, and $k_{pz},k_{dz},k_{p\phi},k_{d\phi},k_{p\theta},k_{d\theta}, k_{p\psi}$ and $k_{d\psi}$ are control parameters. The control law of the SO-SMC-EKF is proposed based on \cite{zheng2014second} as:
\begin{align} 
F_z &=m \frac{c_z\left(\dot{z}_r-\dot{\hat{z}}\right)+\ddot{z}_r+g+d_1+\varepsilon_1 \operatorname{sgn}\left(s_1\right)+\eta_1 s_1}{\cos \phi \cos \theta} \\ 
C_1&=I_{xx}\bigg{ \{ }\frac{c_1}{c_3}\left(\ddot{y}_r-\ddot{\hat{y}}\right)+\frac{c_2}{c_3}\left(\dot{y}_r-\dot{\hat{y}}\right)+\ddot{\phi}_r +d_2 \notag \\
& \ \ \ \ \ \ \ \ \  +\frac{c_4}{c_3}\left(\dot{\phi}_r-\dot{\hat{\phi}}\right)+\frac{1}{c_3}\left[\varepsilon_2 \operatorname{sign}\left(s_2 \right)+\eta_2 s_2\right] \bigg{ \} } \\
C_2&={I_{yy}} \bigg{ \{ }\frac{c_5}{c_7}\left(\ddot{x}_r-\ddot{\hat{x}}\right)+\frac{c_6}{c_7}\left(\dot{x}_r-\dot{\hat{x}}\right)+\ddot{\theta}_r +d_3 \notag \\
& \ \ \ \ \ \ \ \ \ +\frac{c_8}{c_7}\left(\dot{\theta}_r-\dot{\hat{\theta}}\right)+\frac{1}{c_7}\left[\varepsilon_3 \operatorname{sign}\left(s_3\right)+\eta_3 s_3\right] \bigg{ \} }\\
C_3 &=I_{zz}\left[c_\psi\left(\dot{\psi}_r-\dot{\hat{\psi}} \right)+\ddot{\psi}_r+d_4+\varepsilon_4 \operatorname{sign}\left(s_4\right)+\eta_4 s_4\right]
\end{align}
where the sliding surfaces are chosen as:
\begin{align}
    s_1&=c_z(z_r-\hat{z}) + (\dot{z}_r-\dot{\hat{z}}) \\
    s_2&=c_1(\dot{y}_r-\dot{\hat{y}}) + c_2 (y_r-\hat{y})\notag \\
       & \ \ \ \ \ \  + c_3(\dot{\phi}_r-\dot{\hat{\phi}}) + c_4(\phi_r-\hat{\phi})  \\
    s_3&=c_5 (\dot{x}_r-\dot{\hat{x}})+ c_6(x_r-\hat{x}) \notag \\
       & \ \ \ \ \ \ + c_7(\dot{\theta}_r-\dot{\hat{\theta}}) + c_8(\theta_r-\hat{\theta}) \\
    s_4&=c_\psi (\psi_r-\hat{\psi}) + (\dot{\psi}_r-\dot{\hat{\psi}})
\end{align}
in which $\varepsilon_i, \eta_i \ (i=1-4)$, $c_z$, and$ c_k (k=1-8)$ are control parameters, and $d_1,d_2,d_3,d_4$ represent the disturbance terms and are defined as:
\begin{align}
    d_1 &= K_{dz} \dot{\hat{z}}/m \\
    d_2 &= - \dot{\hat{\theta}}\dot{\hat{\psi}} (I_{yy}-I_{zz})/I_{xx}  \\
    d_3 &= - \dot{\hat{\psi}}\dot{\hat{\phi}} (I_{zz}-I_{xx})/I_{yy}  \\
    d_4 &= - \dot{\hat{\phi}}\dot{\hat{\theta}} (I_{xx}-I_{yy})/I_{zz}.
\end{align}
The initial position and the angle values of the UAVs are [0, 0, 0] m and [0, 0, 0] rad, respectively. The system parameters, the observer parameters, and the control parameters are listed in Table~\ref{parameter}.
The simulation is performed in Matlab/Simulink 2023b with the sample time as chosen as $T_s=0.01 (s)$. The simulations are developed on a computer with 2.3 GHz Intel Core i7- 11800H and 16 GB RAM. In addition, the disturbances are white noises, and all of the simulations can be found on this link: \url{https://github.com/aralab-unr/HSMC-EKF-for-Quadrotor-UAVs}. For the simulation scenario, we tested the performance of the quadrotor UAVs in three scenarios, which can be seen in Table~\ref{scenariosimu}. Finally, the results are demonstrated in Figure~\ref{xsetpoint}-Figure~\ref{c3square}. The results are shown in terms of trajectory tracking, the estimated states, and the control laws. The tracking ability and the control law following the Z-axis have shown the effectiveness of the proposed methods with the compared strategies.
\begin{table*}
\centering
\caption{The control performance benchmark under the extended state observer}
{\begin{tabular}{|l|l|l|l|l|l|l|l|} \hline
\bf{System parameter}&\bf{EKF}&\bf{PD-SMC-AC}&\bf{AHSMC-EKF}& \bf{CHSMC-EKF}& \bf{IHSMC-EKF}  & \bf{PID-EKF} & \bf{SO-SMC-EKF}  \\ \hline
$m= 1.96 (kg)$    & $Q= 1e-5$ & $k_{px}=0.1$ & $c_1=0.05$ & $c_1=0.05$  & $c_1=0.05$ & $k_{pz}=10$ & $c_1=0.02, c_2=0.01$ \\ 
$g=9.81 (m/s^2)$    & $R=1e-6$ & $k_{dx}=0.15$ & $c_2=0.05$ & $c_2=0.05$  & $c_2=0.05$ & $k_{dz}=12$ & $c_3=0.2, c_4=0.3$ \\ 
$I_{xx}= 0.00149$ & & $k_{py}=0.1$ & $c_3=1$ & $c_3=1$ & $c_3=0.05$ & $k_{p\phi}=0.6$ & $c_5=0.05, c_6=0.01$\\
$I_{yy}=0.00153$ & & $k_{dy}=0.15$ & $\lambda_1=0.05$ & $\alpha=1.5$ & $c_4=10.05$ & $k_{d\phi}=0.4$ & $c_7=0.2, c_8=0.3$\\
$I_{zz}=0.00532$    & & $c_\phi =3.5$ & $\lambda_2=0.05$ & $K_c=0.5$ & $c_5=3.25$ & $k_{p\theta}=0.6$ & $\eta_1=\eta_4=2$ \\
$K_{dx}= 0.00055670$ & &  $c_\theta=3.5$ & $K_a=0.34$ & $\eta=0.25$& $K_i=0.75$ & $k_{d\theta}=0.4$& $\eta_2=\eta_3=5$  \\
$K_{dy}= 0.00055670$ &  &  $c_\psi=0.5$ & $\eta=0.25$  &  & $\eta=0.25$ & $k_{p\psi}=0.4$ & $\varepsilon_1=1.7, \varepsilon_4=1.2$\\
$K_{dz}= 0.0006354$ & &  $K_\phi=0.4$ &  &  &  & $k_{d\psi}=0.3$ & $\varepsilon_2=\varepsilon_3=1.5$ \\ 
 & &  $K_\theta=0.4$ &  &  & &  & $c_z = 2.5$  \\
 & &  $K_\psi=0.2$ &  &  &  &  & $c_\psi=0.25$ \\\hline
\end{tabular}}
\label{parameter}
\end{table*}	
\begin{table}
\centering
\caption{Simulation scenarios}
{\begin{tabular}{|l|l|l|} \hline
\bf{Scenario}&\bf{Trajectory}&\bf{Time}  \\ \hline
& & \\
Scenario 1 & $[x_r;y_r;z_r,\psi_r]=[12;12;12;0.5]$ & $t=0-15$(s) \\ 
 & & \\ \hline
 & & \\
Scenario 2 & $ x_r= sin(\frac{\pi t}{5})$, $y_r= -1+cos(\frac{\pi t}{5})$  & $t=0-60$(s) \\ 
    & $z_r = \frac{t}{2}$, $\psi_r=0.5$  & \\  
    & & \\ \hline
    & & \\
            & $[x_r;y_r;z_r;\psi_r]=[3;3;3;0.2]$ & $t=0-10$(s) \\  
            & $[x_r;y_r;z_r;\psi_r]=[1.5;3;3;0.2]$ & $t=10-20$(s) \\ 
Scenario 3  & $[x_r;y_r;z_r;\psi_r]=[1.5;1.5;3;0.2]$ & $t=20-30$(s) \\ 
            & $[x_r;y_r;z_r;\psi_r]=[3;1.5;3;0.4]$ & $t=30-40$(s) \\ 
            & $[x_r;y_r;z_r;\psi_r]=[3;3;3;0.4]$ & $t=40-50$(s) \\ 
            & $[x_r;y_r;z_r;\psi_r]=[3;3;0;0.4]$ & $t=50-60$(s) \\  
            & & \\ \hline
\end{tabular}}
\label{scenariosimu}
\end{table}	
\subsection{Scenario 1}
The results under scenario 1 are illustrated from  Figure~\ref{xsetpoint}-Figure~\ref{c3setpoint}, where the position and its estimation are shown in Figure~\ref{xsetpoint}-Figure~\ref{psihsetpoint}, and the control laws are shown Figure~\ref{fzsetpoint}-Figure~\ref{c3setpoint}. 
\begin{figure}[H]
\centering
{\resizebox*{8.5 cm}{!}{\includegraphics{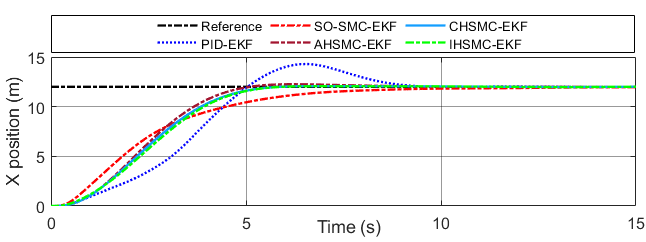}}}\hspace{5pt}
\caption{X-axis tracking trajectory scenario 1.} 
\label{xsetpoint}
\end{figure}
\begin{figure}[H]
\centering
{\resizebox*{8.5 cm}{!}{\includegraphics{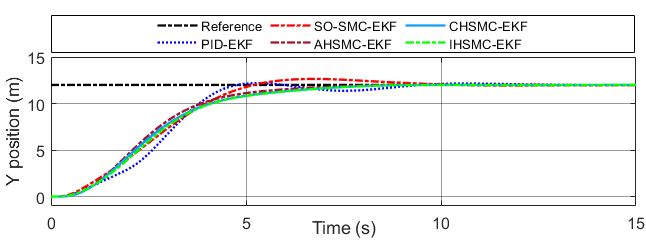}}}\hspace{5pt}
\caption{Y-axis tracking trajectory scenario 1.} 
\label{ysetpoint}
\end{figure}
\begin{figure}[H]
\centering
{\resizebox*{8.5 cm}{!}{\includegraphics{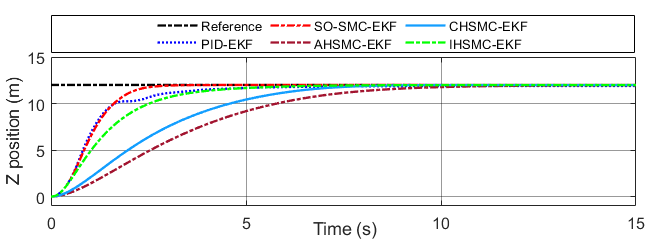}}}\hspace{5pt}
\caption{Z-axis tracking trajectory scenario 1.} 
\label{zsetpoint}
\end{figure}
\begin{figure}[H]
\centering
{\resizebox*{8.5 cm}{!}{\includegraphics{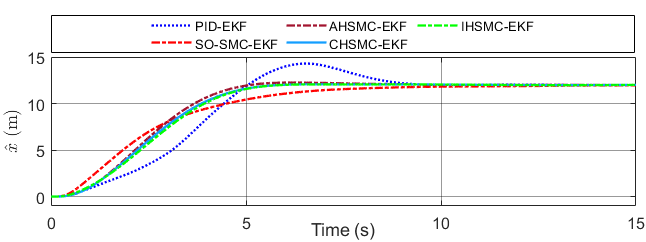}}}\hspace{5pt}
\caption{Estimated X (m) scenario 1.} 
\label{xhsetpoint}
\end{figure}
\begin{figure}[H]
\centering
{\resizebox*{8.5 cm}{!}{\includegraphics{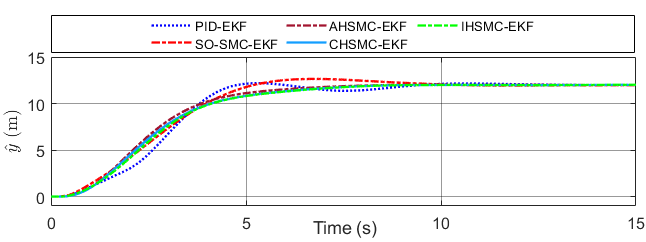}}}\hspace{5pt}
\caption{Estimated Y (m) scenario 1.} 
\label{yhsetpoint}
\end{figure}
\begin{figure}[H]
\centering
{\resizebox*{8.5 cm}{!}{\includegraphics{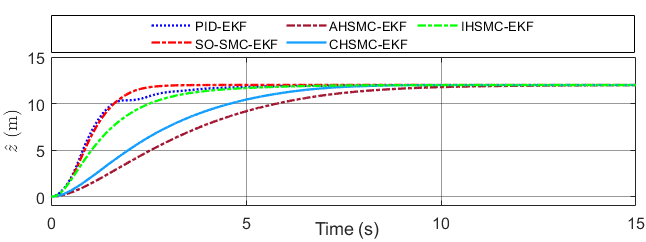}}}\hspace{5pt}
\caption{Estimated Z (m) scenario 1.} 
\label{zhsetpoint}
\end{figure}
\begin{figure}[H]
\centering
{\resizebox*{8.5 cm}{!}{\includegraphics{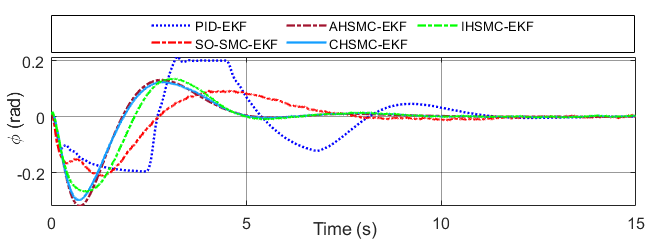}}}\hspace{5pt}
\caption{Roll angle scenario 1.} 
\label{phisetpoint}
\end{figure}
\begin{figure}[H]
\centering
{\resizebox*{8.5 cm}{!}{\includegraphics{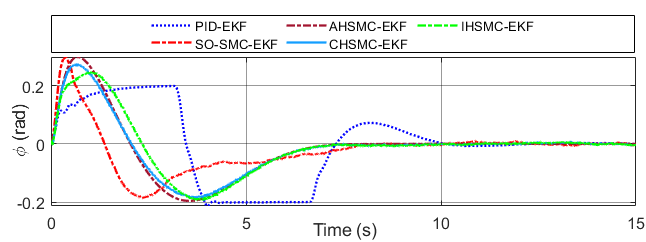}}}\hspace{5pt}
\caption{Pitch angle scenario 1.} 
\label{thetasetpoint}
\end{figure}
\begin{figure}[H]
\centering
{\resizebox*{8.5 cm}{!}{\includegraphics{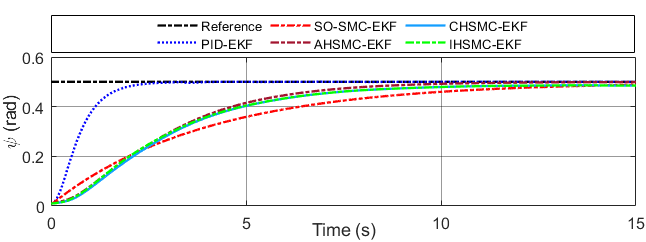}}}\hspace{5pt}
\caption{Yaw angle scenario 1.} 
\label{psisetpoint}
\end{figure}
\begin{figure}[H]
\centering
{\resizebox*{8.5 cm}{!}{\includegraphics{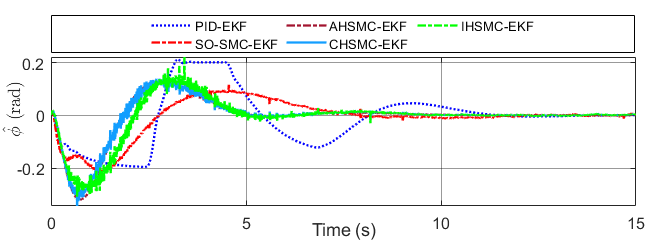}}}\hspace{5pt}
\caption{Estimated roll angle scenario 1.} 
\label{phihsetpoint}
\end{figure}
\begin{figure}[H]
\centering
{\resizebox*{8.5 cm}{!}{\includegraphics{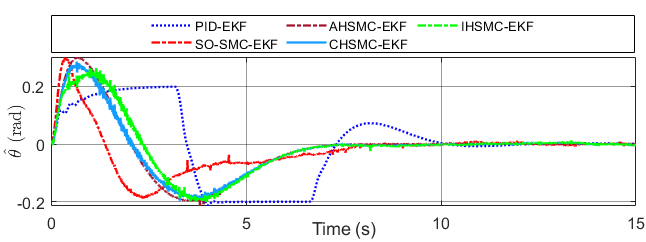}}}\hspace{5pt}
\caption{Estimated pitch angle scenario 1.} 
\label{thetahsetpoint}
\end{figure}
\begin{figure}[H]
\centering
{\resizebox*{8.5 cm}{!}{\includegraphics{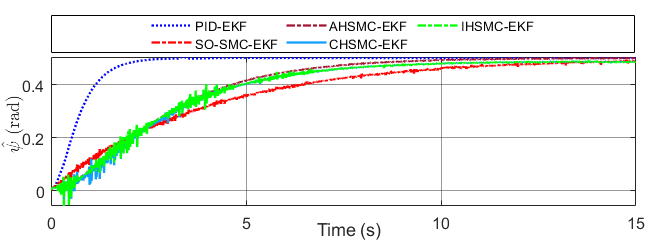}}}\hspace{5pt}
\caption{Estimated yaw angle scenario 1.} 
\label{psihsetpoint}
\end{figure}
\begin{figure}[H]
\centering
{\resizebox*{8.5 cm}{!}{\includegraphics{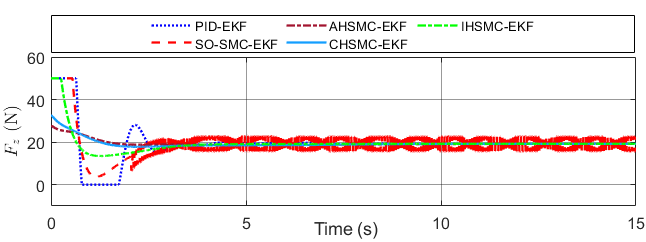}}}\hspace{5pt}
\caption{Control law $F_z$ (N) scenario 1.} 
\label{fzsetpoint}
\end{figure}
\begin{figure}[H]
\centering
{\resizebox*{8.5 cm}{!}{\includegraphics{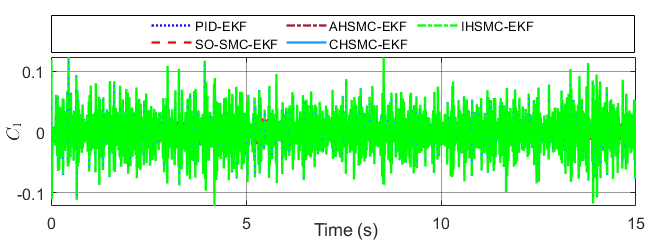}}}\hspace{5pt}
\caption{Control law $C_1$ scenario 1.} 
\label{c1setpoint}
\end{figure}
\begin{figure}[H]
\centering
{\resizebox*{8.5 cm}{!}{\includegraphics{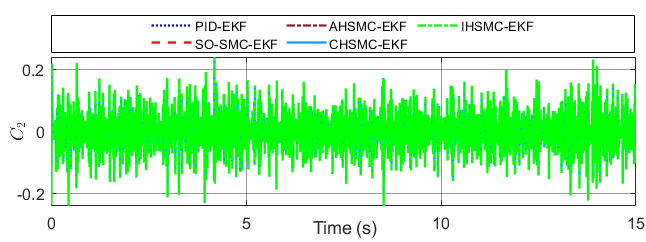}}}\hspace{5pt}
\caption{Control law $C_2$ scenario 1.} 
\label{c2setpoint}
\end{figure}
\begin{figure}[H]
\centering
{\resizebox*{8.5 cm}{!}{\includegraphics{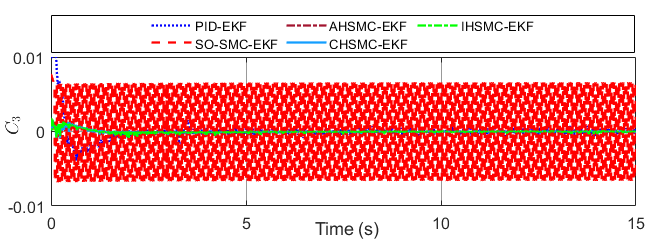}}}\hspace{5pt}
\caption{Control law $C_3$ scenario 1.} 
\label{c3setpoint}
\end{figure}

\subsection{Scenario 2}
The results under scenario 2 are illustrated from  Figure~\ref{xtraject}-Figure~\ref{c3traject}, where the position and its estimation are shown in Figure~\ref{xhtraject}-Figure~\ref{psihtraject}, and the control laws are shown Figure~\ref{fztraject}-Figure~\ref{c3traject}. 
\begin{figure}[H]
\centering
{\resizebox*{8.5 cm}{!}{\includegraphics{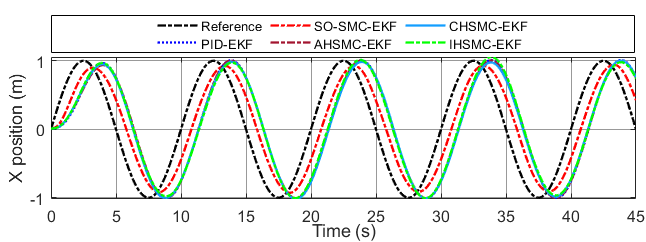}}}\hspace{5pt}
\caption{X-axis tracking trajectory scenario 2.} 
\label{xtraject}
\end{figure}
\begin{figure}[H]
\centering
{\resizebox*{8.5 cm}{!}{\includegraphics{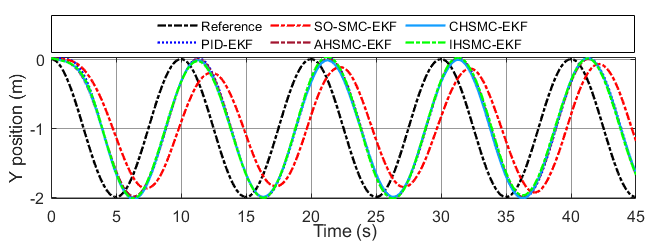}}}\hspace{5pt}
\caption{Y-axis tracking trajectory scenario 2.} 
\label{ytraject}
\end{figure}
\begin{figure}[H]
\centering
{\resizebox*{8.5 cm}{!}{\includegraphics{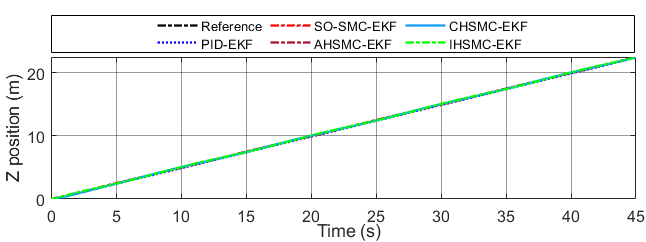}}}\hspace{5pt}
\caption{Z-axis tracking trajectory scenario 2.} 
\label{ztraject}
\end{figure}
\begin{figure}[H]
\centering
{\resizebox*{8.5 cm}{!}{\includegraphics{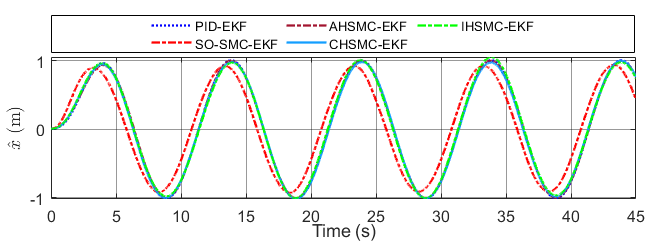}}}\hspace{5pt}
\caption{Estimated X (m) scenario 2.} 
\label{xhtraject}
\end{figure}
\begin{figure}[H]
\centering
{\resizebox*{8.5 cm}{!}{\includegraphics{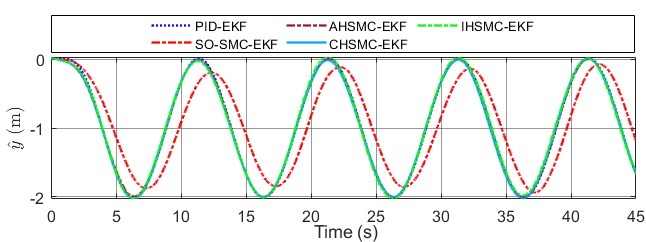}}}\hspace{5pt}
\caption{Estimated Y (m) scenario 2.} 
\label{yhtraject}
\end{figure}
\begin{figure}[H]
\centering
{\resizebox*{8.5 cm}{!}{\includegraphics{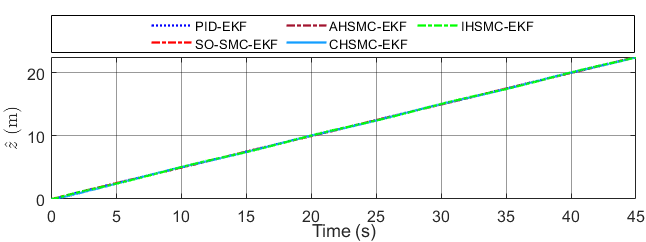}}}\hspace{5pt}
\caption{Estimated Z (m) scenario 2.} 
\label{zhtraject}
\end{figure}
\begin{figure}[H]
\centering
{\resizebox*{8.5 cm}{!}{\includegraphics{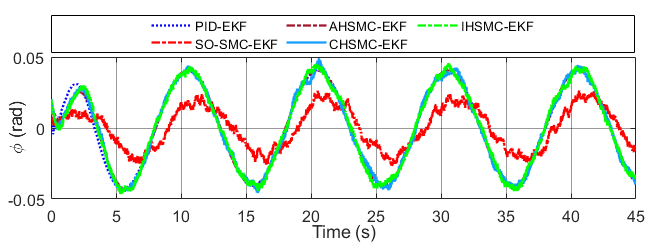}}}\hspace{5pt}
\caption{Roll angle scenario 2.} 
\label{phitraject}
\end{figure}
\begin{figure}[H]
\centering
{\resizebox*{8.5 cm}{!}{\includegraphics{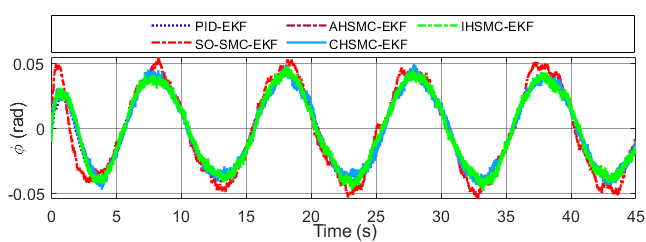}}}\hspace{5pt}
\caption{Pitch angle scenario 2.} 
\label{thetatraject}
\end{figure}
\begin{figure}[H]
\centering
{\resizebox*{8.5 cm}{!}{\includegraphics{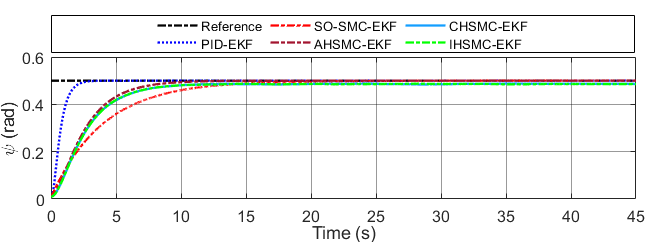}}}\hspace{5pt}
\caption{Yaw angle scenario 2.} 
\label{psitraject}
\end{figure}
\begin{figure}[H]
\centering
{\resizebox*{8.5 cm}{!}{\includegraphics{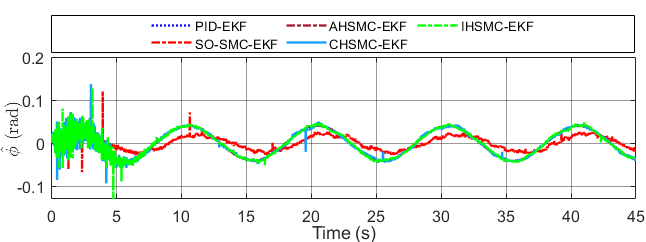}}}\hspace{5pt}
\caption{Estimated roll angle scenario 2.} 
\label{phihtraject}
\end{figure}
\begin{figure}[H]
\centering
{\resizebox*{8.5 cm}{!}{\includegraphics{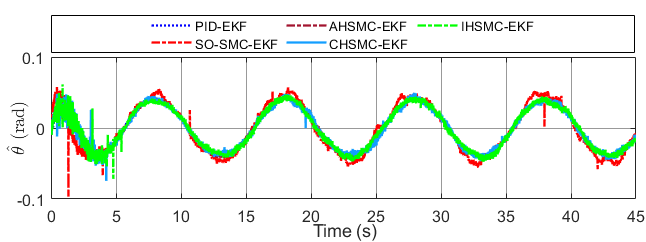}}}\hspace{5pt}
\caption{Estimated pitch angle scenario 2.} 
\label{thetahtraject}
\end{figure}
\begin{figure}[H]
\centering
{\resizebox*{8.5 cm}{!}{\includegraphics{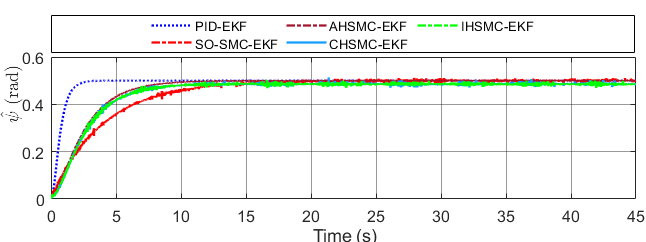}}}\hspace{5pt}
\caption{Estimated yaw angle scenario 2.} 
\label{psihtraject}
\end{figure}
\begin{figure}[H]
\centering
{\resizebox*{8.5 cm}{!}{\includegraphics{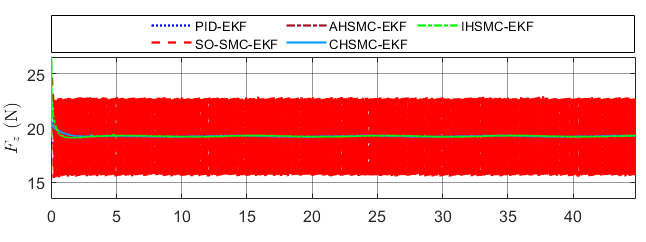}}}\hspace{5pt}
\caption{Control law $F_z$ (N) scenario 2.} 
\label{fztraject}
\end{figure}
\begin{figure}[H]
\centering
{\resizebox*{8.5 cm}{!}{\includegraphics{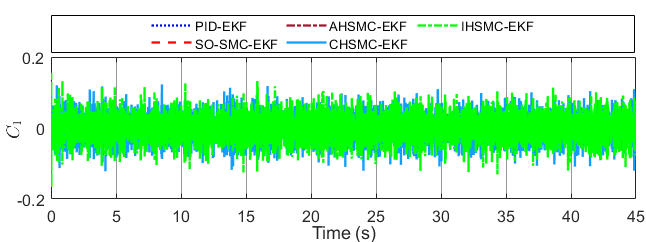}}}\hspace{5pt}
\caption{Control law $C_1$ scenario 2.} 
\label{c1traject}
\end{figure}
\begin{figure}[H]
\centering
{\resizebox*{8.5 cm}{!}{\includegraphics{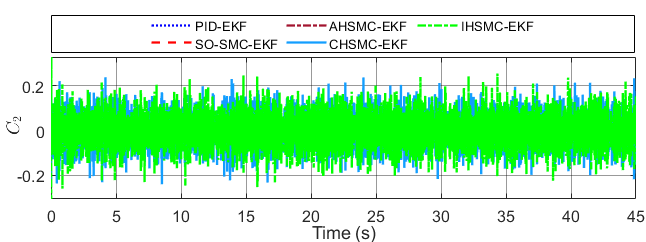}}}\hspace{5pt}
\caption{Control law $C_2$ scenario 2.} 
\label{c2traject}
\end{figure}
\begin{figure}[H]
\centering
{\resizebox*{8.5 cm}{!}{\includegraphics{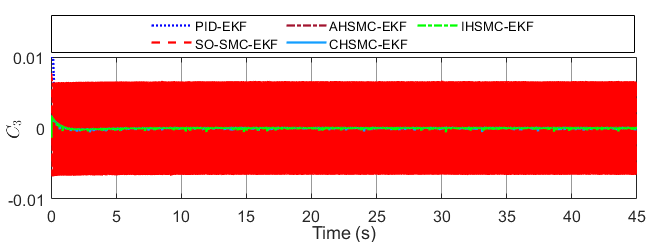}}}\hspace{5pt}
\caption{Control law $C_3$ scenario 2.} 
\label{c3traject}
\end{figure}

\subsection{Scenario 3}
The results under scenario 3 are illustrated in Figure~\ref{xsquare}-Figure~\ref{c3square}, where the position and its estimation are shown in Figure~\ref{xhsquare}-Figure~\ref{psihsquare}, and the control laws are shown Figure~\ref{fzsquare}-Figure~\ref{c3square}. 
\begin{figure}[H]
\centering
{\resizebox*{8.5 cm}{!}{\includegraphics{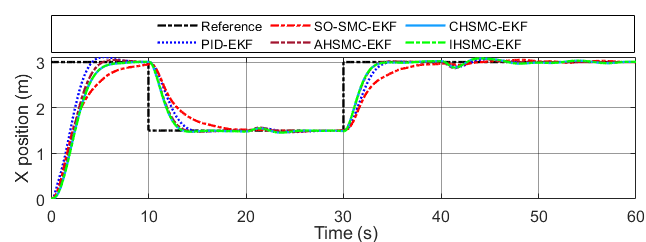}}}\hspace{5pt}
\caption{X-axis tracking trajectory scenario 3.} 
\label{xsquare}
\end{figure}
\begin{figure}[H]
\centering
{\resizebox*{8.5 cm}{!}{\includegraphics{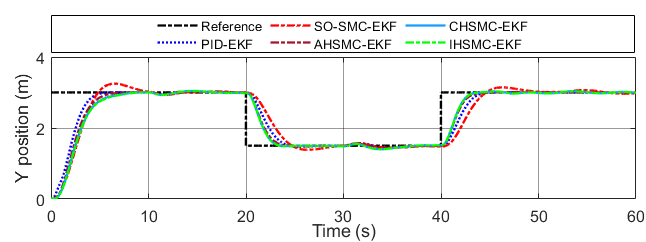}}}\hspace{5pt}
\caption{Y-axis tracking trajectory scenario 3.} 
\label{ysquare}
\end{figure}
\begin{figure}[H]
\centering
{\resizebox*{8.5 cm}{!}{\includegraphics{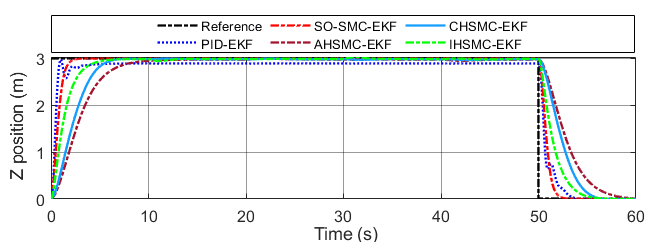}}}\hspace{5pt}
\caption{Z-axis tracking trajectory scenario 3.} 
\label{zsquare}
\end{figure}
\begin{figure}[H]
\centering
{\resizebox*{8.5 cm}{!}{\includegraphics{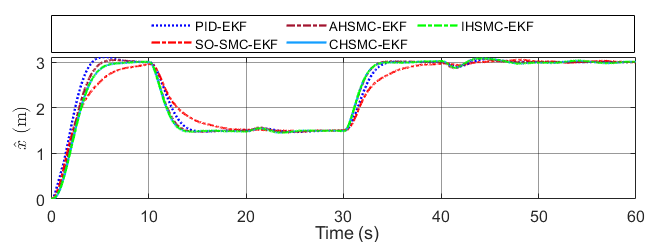}}}\hspace{5pt}
\caption{Estimated X (m) scenario 3.} 
\label{xhsquare}
\end{figure}
\begin{figure}[H]
\centering
{\resizebox*{8.5 cm}{!}{\includegraphics{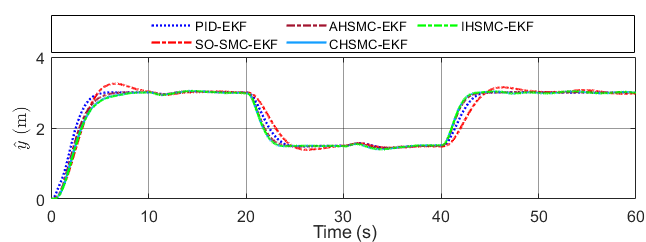}}}\hspace{5pt}
\caption{Estimated Y (m) scenario 3.} 
\label{yhsquare}
\end{figure}
\begin{figure}[H]
\centering
{\resizebox*{8.5 cm}{!}{\includegraphics{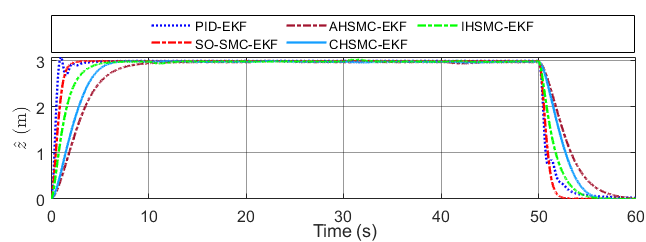}}}\hspace{5pt}
\caption{Estimated Z (m) scenario 3.} 
\label{zhsquare}
\end{figure}
\begin{figure}[H]
\centering
{\resizebox*{8.5 cm}{!}{\includegraphics{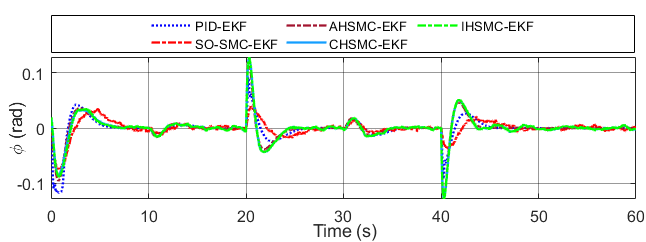}}}\hspace{5pt}
\caption{Roll angle scenario 3.} 
\label{phisquare}
\end{figure}
\begin{figure}[H]
\centering
{\resizebox*{8.5 cm}{!}{\includegraphics{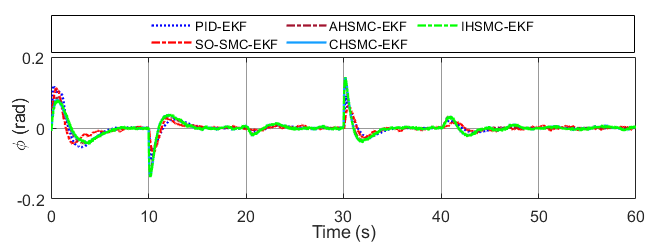}}}\hspace{5pt}
\caption{Pitch angle scenario 3.} 
\label{thetasquare}
\end{figure}
\begin{figure}[H]
\centering
{\resizebox*{8.5 cm}{!}{\includegraphics{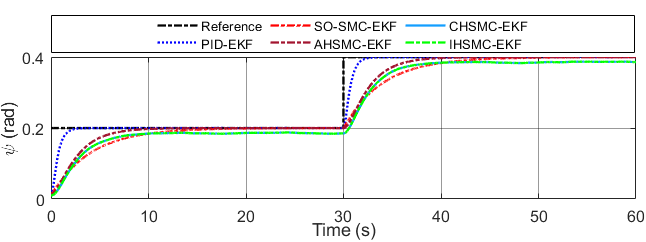}}}\hspace{5pt}
\caption{Yaw angle scenario 3.} 
\label{psisquare}
\end{figure}
\begin{figure}[H]
\centering
{\resizebox*{8.5 cm}{!}{\includegraphics{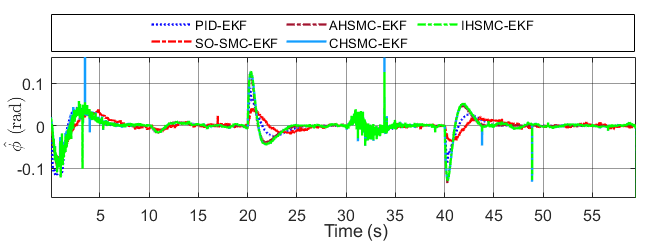}}}\hspace{5pt}
\caption{Estimated roll angle scenario 3.} 
\label{phihsquare}
\end{figure}
\begin{figure}[H]
\centering
{\resizebox*{8.5 cm}{!}{\includegraphics{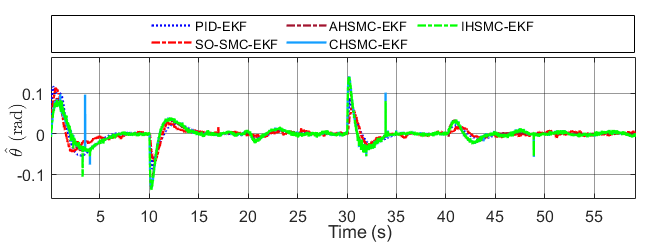}}}\hspace{5pt}
\caption{Estimated pitch angle scenario 3.} 
\label{thetahsquare}
\end{figure}
\begin{figure}[H]
\centering
{\resizebox*{8.5 cm}{!}{\includegraphics{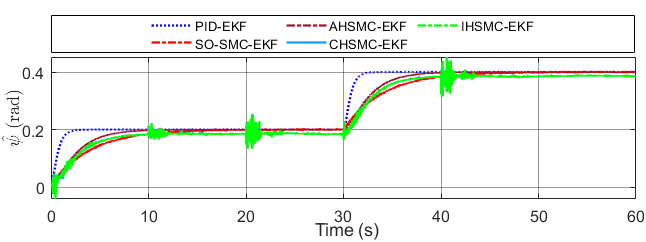}}}\hspace{5pt}
\caption{Estimated yaw angle scenario 3.} 
\label{psihsquare}
\end{figure}
\begin{figure}[H]
\centering
{\resizebox*{8.5 cm}{!}{\includegraphics{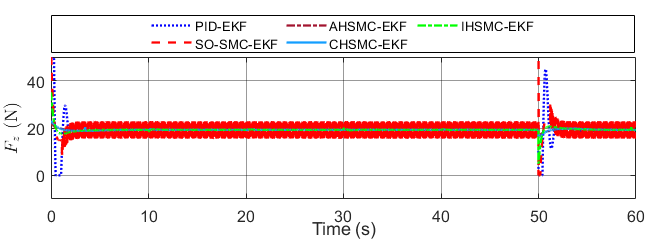}}}\hspace{5pt}
\caption{Control law $F_z$ (N) scenario 3.} 
\label{fzsquare}
\end{figure}
\begin{figure}[H]
\centering
{\resizebox*{8.5 cm}{!}{\includegraphics{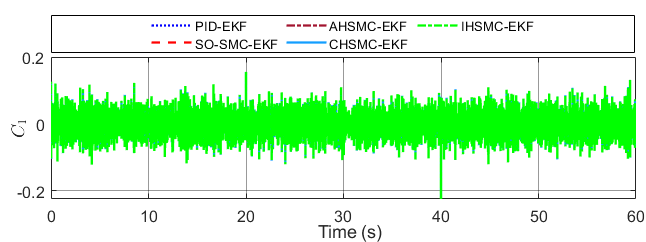}}}\hspace{5pt}
\caption{Control law $C_1$ scenario 3.} 
\label{c1square}
\end{figure}
\begin{figure}[H]
\centering
{\resizebox*{8.5 cm}{!}{\includegraphics{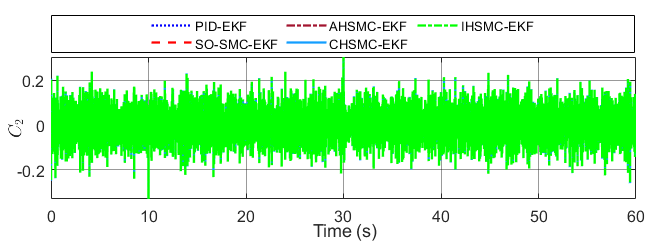}}}\hspace{5pt}
\caption{Control law $C_2$ scenario 3.} 
\label{c2square}
\end{figure}
\begin{figure}[H]
\centering
{\resizebox*{8.5 cm}{!}{\includegraphics{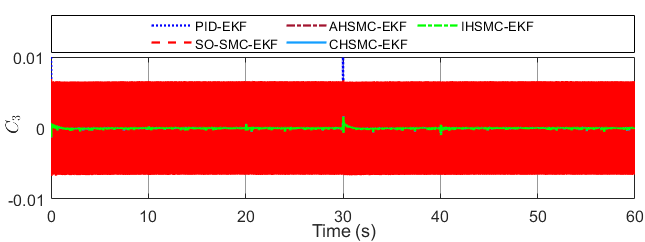}}}\hspace{5pt}
\caption{Control law $C_3$ scenario 3.} 
\label{c3square}
\end{figure}
In general, all the proposed methods can ensure that the quadrotor UAVs track the desired trajectory without compromising their stability, even in the presence of white noise. These tracking abilities and stability show the effectiveness of the EKF and the PD-SMC attitude controller. In terms of position control, it is observed that the  HSMC outperforms both the PID-EKF and SO-SMC-EKF approaches. Notably, the PID-EKF method exhibits position overshoot in Scenario 1, while the SO-SMC method demonstrates a smaller amplitude along the X and Y axes in comparison to the desired trajectory. The HSMC effectively mitigates these issues. Furthermore, in analyzing control forces along the Z-axis, it is found that the proposed HSMC, utilizing the saturation sign function, effectively disregards the influence of switch control laws, thereby ensuring consistent UAV tracking performance. This contrasts with the limitations encountered by the SO-SMC method. The observed fluctuations in control prove the efficacy of the HSMC methodology.
\section{Conclusion}
In this paper, we present a novel approach to control quadrotor UAVs. First, the EKF is proposed to deal with external disturbances and measure noises. Second, a PD-SMC attitude controller is proposed to create the reference roll and pitch angles for the UAVs, and a class of hierarchical sliding mode control is designed to ensure the tracking ability and stability of the UAVs. The main conclusions are summarized as follows: 1. All the estimated state variables by the EKF converge to the real states under the influence of lump disturbances; 2. The tracking ability of the UAVs is ensured in the case of the reference, which is a desired point, a desired trajectory, or reference values suddenly changed in a different moment; 3. The robustness of the designed controller was verified based on the rigorous simulation scenarios and the comparison to other existing methods. In the future,  utilizing the promising Lyapunov function of the HSMC, we will consider the development of an adaptive HSMC to deal with the changes in the UAV's mass or parameters. Furthermore, the controller can be updated to Lyapunov-based model predictive control to further address system constraints for the UAVs, along with control laws $F_z$, $C_1$, $C_2$, and $C_3$.
\appendix
The details of the linearized matrix $\mathbf{A}({\mathbf{X}}_{k},\mathbf{U}_{k})$,$ \mathbf{B}({\mathbf{X}}_{k},\mathbf{U}_{k}) $ are indentifed as:
\begin{align}
   & \mathbf{A}({\mathbf{X}}_{k},\mathbf{U}_{k}) = \begin{pmatrix}
 \mathbf{0}_{6 \times 6} & \mathbf{I}_{6 \times 6} \\
 \mathbf{0}_{6 \times 6} &  \mathbf{A_1}({\mathbf{X}}_{k},\mathbf{U}_{k})
 \end{pmatrix} \\
 & \mathbf{B}({\mathbf{X}}_{k},\mathbf{U}_{k}) = \begin{pmatrix}
 \mathbf{0}_{6 \times 6} \\
 \mathbf{B_1}({\mathbf{X}}_{k},\mathbf{U}_{k}) \\
 \end{pmatrix}
\end{align}
where:
\begin{align}
   & \mathbf{A_1}({\mathbf{X}}_{k},\mathbf{U}_{k}) = \begin{pmatrix}
 \mathbf{A_{11}} & \mathbf{0}_{3 \times 3}\\
 \mathbf{0}_{3 \times 3} &  \mathbf{A_{22}}
 \end{pmatrix}\\
 & \mathbf{B_1}({\mathbf{X}}_{k},\mathbf{U}_{k}) = \begin{pmatrix}
 \mathbf{B_{11}} & \mathbf{0}_{3 \times 3}\\
 \mathbf{0}_{3 \times 1} &  \mathbf{B_{22}}
 \end{pmatrix}
 \end{align}
 with
\begin{align}
   & \mathbf{A_{11}} = \begin{pmatrix}
  -\frac{K_{dx}}{m} & 0 & 0 \\ 
0 & -\frac{K_{dy}}{m}& 0\\ 
0 & 0 & -\frac{K_{dz}}{m}
 \end{pmatrix} \\
 & \mathbf{A_{22}} = \begin{pmatrix}
 0 & \frac{{\dot{\psi}}\,\left({I_{yy}}-{I_{zz}}\right)}{{I_{xx}}} & \frac{{\dot{\theta}}\,\left({I_{yy}}-{I_{zz}}\right)}{{I_{xx}}}\\ 
 \frac{{\dot{\psi}}\,\left({I_{zz}}-{I_{xx}}\right)}{{I_{yy}}} & 0 & \frac{{\dot{\phi}}\,\left({I_{zz}}-{I_{xx}}\right)}{{I_{yy}}}\\ 
\frac{{\dot{\theta}}\,\left({I_{xx}}-{I_{yy}}\right)}{{I_{zz}}} & \frac{{\dot{\phi}}\,\left({I_{xx}}-{I_{yy}}\right)}{{I_{zz}}} & 0
 \end{pmatrix} \\
& \mathbf{B_{11}} = \begin{pmatrix}
 \frac{\sin\left(\phi \right)\,\sin\left(\psi \right)+\cos\left(\phi \right)\,\cos\left(\psi \right)\,\sin\left(\theta \right)}{m} \\ 
\frac{\cos\left(\phi \right)\,\sin\left(\psi \right)\,\sin\left(\theta \right)-\cos\left(\psi \right)\,\sin\left(\phi \right)}{m} \\
\frac{\cos\left(\phi \right)\,\cos\left(\theta \right)}{m}
 \end{pmatrix} \\
 & \mathbf{B_{22}} = \begin{pmatrix}
\frac{1}{{I_{xx}}} & 0 & 0\\ 
0 & \frac{1}{{I_{yy}}} & 0\\ 
0 & 0 & \frac{1}{{I_{zz}}} 
 \end{pmatrix}.
 \end{align}
\bibliographystyle{IEEEtran}
\bibliography{main}

% biography section
\begin{IEEEbiography}[{\includegraphics[width=1.1in,height=1.1in,clip,keepaspectratio]{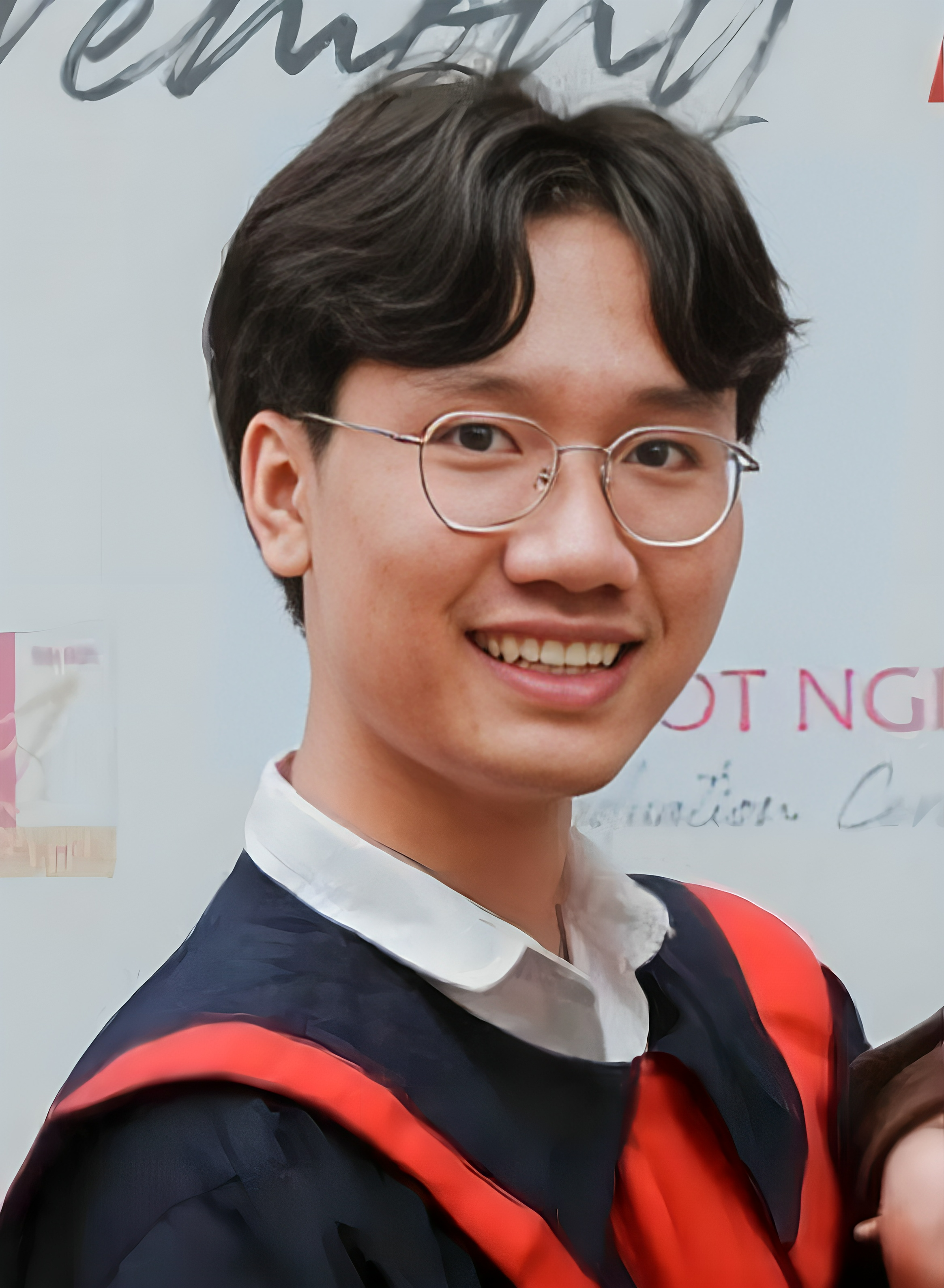}}]{Van Chung Nguyen} (Student Member, IEEE) received
a B.S. degree in Control \& Automation Engineering in 2023 from Hanoi University of Technology, Hanoi, Vietnam.

Mr. Nguyen is currently a research assistant at the Advanced Robotics and Automation (ARA) Lab, and pursuing a Ph.D. degree at the Department of Computer Science and Engineering, University of Nevada, Reno.  His research interests include robotics, control, optimization, and navigation systems.
\end{IEEEbiography}
% \vspace{-50pt}

 \begin{IEEEbiography}[{\includegraphics[width=1in,height=1in,clip,keepaspectratio]{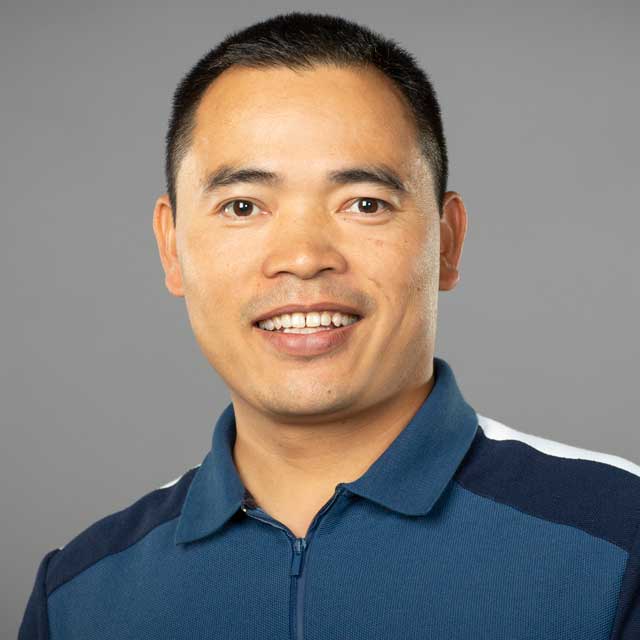}}]
{Hung La} (Senior Member, IEEE) received
a Ph.D. degree in electrical and computer engineering from Oklahoma State University, Stillwater, OK, USA, in 2011.

Dr. La is the Director of the Advanced Robotics and Automation Lab, and an Associate Professor of the Department of Computer Science and Engineering, University of Nevada, Reno, NV, USA. From 2011 to 2014, he was a Post-Doctoral research fellow and then a Research Faculty Member at the Center for Advanced Infrastructure and Transportation, Rutgers University, Piscataway, NJ, USA. He has authored over 160 papers published in major journals, book chapters, and international conference proceedings. His current research interests include robotics, deep learning, and AI.
Dr. La is the recipient of the 2019 NSF CAREER award, and the 2014 ASCE Charles Pankow Award for the Robotics Assisted Bridge Inspection Tool (RABIT), 12 best paper awards/finalists, and the best presentation award in international conferences (e.g., ICRA, IROS, ACC, SSRR). Dr. La is an Associate Editor of the IEEE International Conference on Robotics and Automation (ICRA), and the IEEE International Conference on Systems, Man, and Cybernetics (SMC). He is also an Associate Editor of several journals including IEEE Transactions on Human-Machine Systems, Robotica, and Frontier in Robotics \& AI. He was the Guest Editor of the International Journal of Robust and Nonlinear Control. 

\end{IEEEbiography}

\end{document}